\documentclass[12pt]{article}

\usepackage{graphicx} 
\usepackage[utf8]{inputenc}
\usepackage{fullpage, amsmath, amsfonts, amsthm, color, graphicx, amssymb, nicefrac, bbm}
\usepackage{algpseudocode}
\usepackage{hyperref}

\usepackage{authblk}
\usepackage{comment}

\usepackage{enumerate}

\newtheorem{theorem}{Theorem}
\newtheorem{lemma}{Lemma}

\theoremstyle{definition}
\newtheorem{algorithm}{Algorithm}
\newtheorem{definition}{Definition}
\newtheorem{remark}{Remark}
\newtheorem{problem}{Problem}

\newcommand{\hidden}[1]{}

\newcommand{\Span}{\text{span}}

\renewcommand{\Re}{\text{Re}}

\newcommand{\gap}{\text{gap}}

\newcommand{\Bump}{\text{Bump}}
\newcommand{\Cut}{\text{Cut}}
\newcommand{\Sat}{\text{Sat}}
\renewcommand{\Bar}{\text{Bar}}

\newcommand{\rn}{{\mathbb R^n}}
\newcommand{\tn}{{\mathbb T^n}}
\newcommand{\nn}{{\mathcal N^n}}

\graphicspath{ {./images/} }

\title{{A quantum analogue of convex optimization }}
\date{\today}
\author{Eunou Lee\thanks{ \href{mailto:eunoulee@kias.re.kr}{eunoulee@kias.re.kr}}}
\affil{Korea Institute for Advanced Study}

\begin{document}

\maketitle

\begin{abstract}
Convex optimization is the powerhouse behind
the theory and practice of optimization.
We introduce a quantum analogue of unconstrained convex optimization:
find the minimum
eigenvalue of a Schr\"odinger operator
$h  =   -\Delta   + V $
with convex 
potential $V:\rn \rightarrow \mathbb R_{\ge 0}$ such 
that $V(x)\rightarrow\infty $ as $\|x\|\rightarrow\infty$.
For this problem, we present an efficient quantum algorithm
that computes the minimum eigenvalue of $h$ up to error 
$\epsilon$ in polynomial time in $n$, $1/\epsilon$, 
and parameters that depend on $V$. 
Adiabatic evolution of the ground 
state
is used as a key subroutine, 
which we analyze with novel techniques that allow 
us to focus on the low-energy space.
We apply our algorithm
to give the first known polynomial-time algorithm 
for finding the lowest frequency of an $n$-dimensional convex drum,
or mathematically, the minimum eigenvalue 
of the Dirichlet Laplacian on an $n$-dimensional 
region that is defined by $m$ linear constraints in
polynomial time in $n$, $m$, $1/\epsilon$ and the radius $R$ 
of a ball encompassing the region.
\end{abstract}

\setcounter{tocdepth}{2}

\tableofcontents

\section*{Part I: Introduction and results}
\addcontentsline{toc}{section}{Part I: Introduction and results}
\section{Introduction}

Convex optimization is the infrastructure for
the theory and practice of optimization.
The study of minimization of convex functions has led to a wide range of tools that are
essential even for solving nonconvex and combinatorial problems~\cite{BV04, WS11}.
For optimization problems over quantum states, however,
such a versatile framework is yet to be developed.

In this work, we introduce 
a quantum analogue of 
unconstrained convex optimization,
which is
Schr\"odinger convex optimization:
find the minimum
eigenvalue of a Schr\"odinger operator
\begin{align*}
    h \ =  \ -\Delta \  + \ V \ = \ - \sum_{i\in[n]} \frac{\partial^2}{\partial x_i^2}  \ + \ V
\end{align*}
with convex 
potential $V:\rn \rightarrow \mathbb R_{\ge 0}$
such that $V(x)\rightarrow\infty$
as $\|x\|\rightarrow\infty$.
For this problem, we present an efficient quantum algorithm,
called the Adiabatic Schr\"odinger Convex Optimization Algorithm~(ASCOA),
that computes the minimum eigenvalue of $h$ up to error 
$\epsilon$ in polynomial time in $n$, $1/\epsilon$, 
and parameters that depend on $V$. 
Adiabatic evolution of the ground 
state~\cite{BF28,FGGS00}
is used as a key subroutine
that
we analyze with novel techniques that allow 
us to focus on the low-energy space.
We apply the ASCOA to give the first known polynomial-time algorithm 
for finding the lowest frequency of an $n$-dimensional convex drum 
up to error $\epsilon$.
More precisely, we compute the minimum eigenvalue 
of the Dirichlet Laplacian on an $n$-dimensional 
polygon in
polynomial time in $n$, $1/\epsilon$,
 the number of linear constraints $m$ defining the polygon,
and the radius $R$ 
of a ball encompassing the region.
Existing algorithms,
such as 
Finite Element Methods,
solve the drum problem
in 
exponential time in $n$,
exhibiting the curse of dimensionality.

The Schr\"odinger operator $h$ acts on a function
$\psi :\rn\rightarrow \mathbb C$ to give another function $h\psi:\rn\rightarrow\mathbb C$ such that
\begin{align*}
    (h\psi)(x) = -\sum_{i\in[n]}\frac{\partial ^2}{\partial x_i^2}\psi(x) + V(x)\psi(x).
\end{align*}
The minimum eigenvalue problem is to compute the
smallest number $\lambda_0\in\mathbb R$ for which there
exists $\psi\ne0$ such that $h\psi = \lambda_0\psi$. 
Since the eigenvectors of $h$ form a complete set of basis,
we have an equivalent characterization
\begin{align*}
    \lambda_0 = \min_{\langle \psi|\psi\rangle = 1}
    \langle \psi| h\psi\rangle , \qquad \text{where} \qquad \langle \phi|\psi\rangle = \int_\rn \overline{\phi(x)}\psi(x) \text d x.
\end{align*}
Physically, the wavefunction $\psi$ 
such that $\langle\psi|\psi\rangle =1 $ represents
a quantum particle in $\rn$.
The total mechanical energy of the particle is
$\langle \psi|h\psi\rangle $,
which is the sum of the kinetic energy 
$-\langle \psi|\Delta\psi\rangle$ and the the 
potential energy
$\langle \psi|V\psi\rangle$.
Our goal is to find the minimum mechanical
energy that a quantum 
particle can have under the convex potential $V$.

We consider 
Schr\"odinger convex optimization to be
a quantum analogue of convex optimization because:
$i)$ the objective $h$ is convex, 
$ii)$ its spectrum is quantized, 
and $iii)$ an efficient quantum algorithm 
exists for the problem.
To expand,
$h$ is convex 
since it is the sum of 
$-\Delta$ and $V$,
each of which
is convex in the Fourier and the position domain
respectively
(see Section~\ref{sec:discussion} 
for a more detailed discussion.)
Even though each of $-\Delta$ and $V$
has a continuous spectrum,
their sum $h$
has a purely discrete, i.e. quantized, spectrum.
The quantization of
allowed energy levels of
a Schr\"odinger operator
with the Coulomb potential
is the phenomena that confounded
physicists a century ago,
prompting the birth of quantum physics.
Finally, an efficient quantum algorithm
for the problem completes the analogy,
as the importance of
classical convex optimization comes from 
its efficient solvability.

It is intuitive to believe that 
an efficient quantum algorithm
should exist
for Schr\"odinger convex optimization.
Its classical limit, 
minimization of the mechanical energy
of a classical particle,
is efficiently solved by
simulating friction:
a classical particle 
on a convex surface
loses its energy
to the minimum
due to friction.
Accelerated Gradient Descent~\cite{nest83}
is an algorithmic implementation of
this physical idea~\cite{WWJ16}.
Likewise, 
we can expect
an efficient quantum algorithm
solving Schr\"odinger 
convex optimization,
for example
by simulating
quantum dissipation~\cite{Lin25, KBGKE11}.

Interestingly enough, our algorithm is based 
on simulating adiabatic Hamiltonian evolution,
rather than on simulation of dissipation.
The most challenging part of 
analyzing an adiabatic algorithm
is to establish a lower bound
of the spectral gap.
Fortunately for us, 
a tight lower bound 
for
the spectral gap of
a Schr\"odinger operator
on a bounded domain 
is known to be 
inverse polynomial 
in the diameter of the domain
by the Fundamental Gap Theorem~\cite{AC10}.
Therefore, we want to 
$i)$ minimize a Schr\"odinger 
operator
over the unbounded domain $\rn$
$ii)$
by simulating 
a discretized qubit Hamiltonian
$iii)$
whose gap we analyze 
using results 
for 
bounded domains.
We relate the three operators 
by the tools that we develop 
that allow us to ignore high-energy 
parts of a Hamiltonian.
We define a low-energy truncation of $h$,
which represents the operator up to its low-energy sector
with error $\epsilon$ (Definition~\ref{def:truncation}). 
We show that if two Hamiltonians admit similar truncations,
then their spectral gaps and the ground energies are 
close to each other.
For the analysis, we construct a 
series of intermediate Hamiltonians on different domains
that are treated in a modularized manner 
and connected
via the truncation techniques.

As an application of the ASCOA, 
we give the first known polynomial-time algorithm 
for finding the minimum frequency of an 
$n$-dimensional convex drum
specified by linear constraints.
Mathematically, the problem is 
to compute the minimum eigenvalue of the Laplacian under 
the Dirichlet boundary condition on an $n$-dimensional polytope.
The runtime of classical solvers such as finite element methods~\cite{Babuska1989DL}
grows exponentially to $n$, exhibiting the curse of dimensionality.
Information theoretic arguments~\cite{Papageorgiou2007DL} show that 
no classical algorithm runs in time polynomial in both $n$ and $1/\epsilon$.

As far as we are aware,
only
three previous works~\cite{PP14, CHPZ25, DZPL25}
exhibit 
polynomial time quantum algorithms 
for optimization of quantum objective
with a provable guarantee.
Efficient quantum
algorithms
based on simulating dissipation
find
local minima of qubit Hamiltonians~\cite{CHPZ25}
and 
the ground energy
of quadratic fermionic Hamiltonians~\cite{DZPL25}.
More closely related to our work is a quantum algorithm 
computing the ground energy 
of a Schr\"odinger operator on a bounded domain $[0,1]^n$ 
with the Dirichlet boundary condition~\cite{PP14}.
This algorithm prepares the ground state of an initial Hamiltonian,
and then repeatedly measures in the energy basis
of a slowly varying Hamiltonian.
They show the spectral gap of the time-dependent Hamiltonian
is large, by employing
the Fundamental Gap Theorem~\cite{AC10},
which we also adapt in this work.

Other related works concern with 
noncommutative convex optimization.
Geodesic convex optimization~\cite{Rapcsak1991,ZhangSra2016,AllenZhu2018,BFGOWW19}
also studies convex  
objectives in 
noncommutative directions,
but our case
differs in that
the domain is infinite-dimensional.
Another notion of noncommutative convex optimization 
appears in Noncommutative Polynomial Optimization~(NcPO)~\cite{Helton02,HeltonMcCullough04,BKP16,NPA08}.
When the potential $V$ is a polynomial, 
our formulation can be viewed as an
infinite-dimensional instantiation of NcPO.

\section{Problem formulation and results}

\subsection{Quantum optimization of quantum objectives}
Our objective function $h$ is defined on the function space
\begin{align*}
    L^2(\rn):= \left\{\psi:\rn\rightarrow\mathbb C \ | \  \int_{\rn} | \psi|^2 \ \text d x <\infty \right\}.
\end{align*}
An optimization problem is often stated as: find the best element from a given set.
It is unclear 
what it means to find 
the best element 
over the inifinite-dimensional space $L^2(\rn)$,
using 
a quantum computer hosting only a finite number of qubits.

The infinity of the domain is
already a familiar issue in classical optimization.
For instance, we know $\min_{x\in\mathbb R} x^2 =0$ 
at $x = 0$ even though $\mathbb R$ is infinite.
More generally, 
if we are given a certificate for $f(x^*)$,
we know that $\min_x f(x)\le  f(x^*)$.
This idea of a certificate for optimization
can be formalized as follow.

\begin{definition}[generalization of~\cite{LP24}]\label{def:optimization}
    Given an objective function $f:X \rightarrow\mathbb R$, a pair of algorithms $(p,v)$ minimizes $f$ to a value $\alpha$ up to error $\epsilon$, if the following conditions hold.
\begin{enumerate}
    \item The algorithm $v$ correctly verifies $\min_X f\le \alpha$:
    \begin{enumerate}
        \item  If $\text{min}_{x\in X}f(x) \le \alpha$:  $\exists$ $|w\rangle$ on finitely many (qu)bits such that $v(|w\rangle) =1$ w.p. $\ge 2/3$.
        \item  If $\text{min}_{x\in X}f(x) \ge \alpha +\epsilon$:  $\forall$ $ |w\rangle$ on finitely many (qu)bits, we have $v(|w\rangle) =1$ w.p. $\le 1/3$.
    \end{enumerate}
    \item The algorithm $p$, w.p. $\ge 2/3$, outputs a finite length $|w'\rangle$ such that $v(|w'\rangle)=1 $.
\end{enumerate}
\end{definition}
In this paper, we find the minimum of objective functions in polynomial time in the following sense.

\begin{definition}
    Given an objective function $f:X \rightarrow\mathbb R$, a pair of algorithms $(p,v)$ finds the minimum of $f$ up to error $\epsilon$ in time $T$, if $(p,v)$ minimizes $f$ to the value $\min_{x\in X} f(x)$ 
    up to error $\epsilon$
    and both $p$ and $v$ halt in time $T$.
\end{definition}

Whether 
$p$, $v$, and $|w\rangle$ are classical or quantum is left unspecified by the definition;
all classical/quantum combinations are allowed.
In our case, $p, v$ and $|w\rangle$ are all quantum.

\subsection{Schr\"odigner convex optimization}
The main problem of this paper is to 
find the minimum nonzero eigenvalue $\lambda_0$ of
a Schr\"odinger operator $h = -\Delta + V$
where $V(x)\rightarrow\infty$ 
as $\|x\|\rightarrow\infty$, and $V$ is convex.

\begin{problem}[Schr\"odinger]\label{prob:schr}
    Given a Schr\"odinger operator $h=-\Delta + V$ with a
    smooth convex function $V:\rn\rightarrow\mathbb R_{\ge 0}$ such that $V(x)\rightarrow\infty$ as $\|x\|\rightarrow\infty$,
     \begin{align*}
 \text{minimize}& \qquad \lambda_0 \\
\text{subject to}& \qquad \psi:\rn\rightarrow\mathbb C,\\
                &\qquad \lambda_0 \psi =   h\psi, \\
               & \qquad \psi\ne 0.
    \end{align*}    
\end{problem}

We have an equivalent formulation in terms of energy,
since the eigenfunctions of $h$ 
forms a complete basis for $L^2(\rn)$.
We define the energy functional as follows.

\begin{definition}[Energy]\label{def:energy}
    Let $h$ be a self-adjoint operator on a Hilbert space $\mathcal H$ with operator domain $\mathcal D(h)\subseteq \mathcal H$.
The \textbf{energy} of a nonzero vector $\psi \in \mathcal{D}(h)$ with respect to $h$ is denoted
\[
h[\psi] := \frac{\langle \psi | h  \psi \rangle}{\langle \psi | \psi \rangle}.
\]
\end{definition}
The Hilbert space $\mathcal H$ is $L^2(\rn)$ for the Schr\"odinger problem.
The operator domain $\mathcal D(h)\subseteq \mathcal H$ is the subset on which $h\psi \in L^2(\rn)$ is defined in a self-adjoint way. 
For more explanation, see Section~\ref{sec:prelim}.
The following is an equivalent formulation of Problem~\ref{prob:schr} in terms of energy.

\begingroup
\let\savedtheproblem\theproblem     
\addtocounter{problem}{-1}          
\renewcommand{\theproblem}{\savedtheproblem$'$} 
\begin{problem}[Schr\"odinger]
    Given a Schr\"odinger operator $h=-\Delta + V$ with a
    smooth convex function $V:\rn\rightarrow\mathbb R_{\ge 0}$ such that $V(x)\rightarrow\infty$ as $\|x\|\rightarrow\infty$,
     \begin{align*}
 \text{minimize}& \qquad h[\psi].    \end{align*}
\end{problem}
\endgroup

There is an ambiguity as to how $V$ is given to us.
We assume that we have access to a circuit
that computes $V(x)$ when $\|x\|_\infty \le L/2$.
The ASCOA efficiently solves Problem~\ref{prob:schr} for 
$V$ satisfying certain conditions.

\begin{theorem}[main]\label{thm:main}
    Algorithm~\ref{alg:ASCOA} computes the lowest eigenvalue of $h = -\Delta + V$ within error $\epsilon_0$ with probability $\ge 2/3$ in polynomial time in $n, 1/\epsilon_0,L,c$,
    if $V$
    satisfies the 
    following conditions, where 
    $B^q_d:=\{x\in\rn \ | \  \|x\|_q\le d\}$ and 
     $V_E^{-1}:= V^{-1}([0,E])$:
    \begin{enumerate}
        \item  \emph{Purely discrete spectrum and convexity:} we have $V(x)\rightarrow\infty$ as $\|x\|\rightarrow\infty$ and $V$ is convex on $B^2_{r+1}$.
        \item  \emph{``Bowl-shaped'' in $B^\infty_{L/2}$:} there exist $a\le b\le c$, and $1 \le r < L/2$ such that
        \begin{align*}
            B^2_1\subset V^{-1}_{a } \subset V^{-1}_{b}\subset B^2_r \subset B^2_{r+1} \subset V^{-1}_{c}\subset V^{-1}_{c+1} \subset B^\infty_{L/2}
        \end{align*}  
         where
        \begin{enumerate}[(i)]
            \item $b = \Theta(E_0^7/\sigma^6)$
            \item $E_0=10(n(n+3)\pi^2 +a) = \Theta(n^2+a)$
            \item $\sigma = \Theta(\epsilon_0/r^4 E_0^{1.5})$
            \item $\epsilon_0<1 < r,E_0,c$. 
        \end{enumerate}
        \item \emph{Bounded high-order derivatives:} the potential $V$ is smooth and 
        \begin{align*}
            \max \left\{ { \frac{1}{l}\log|\partial^l_j V(x)| }  \  \bigg|  \ {l\in[m], j\in[n], x\in [-L/2,L/2]^n }\right\}  \le \log p(m)
        \end{align*} for some polynomial $p$.
        \item \emph{Access to $V$:} there exists a circuit that computes $V(x)$ for all $ x \in B^\infty_{L/2} $ with a negligible error in $n, 1/\epsilon_0,  L, c$ in polynomial time in the same parameters.
    \end{enumerate}
\end{theorem}
\begin{proof}
    See Section~\ref{sec:proof_of_main_thm}.
\end{proof}
Note that $V$ is not required to be convex only in $B^2_{r+1}$,
and also, not every convex $V$ satisfy the conditions.

\subsection{The drum problem}
We provide an example to demonstrate that
Theorem~\ref{thm:main} is not void, and indeed useful.
The example is to compute the minimum eigenvalue of the Laplacian
on an $ n$-dimensional domain $\Omega$.
We impose the Dirichlet boundary condition,
which enforces that
\begin{align*}
    \psi(x)  &= 0 \qquad\qquad \forall x\in\partial \Omega.
\end{align*}
We assume that we are given 
\begin{align*}
    \Omega = \{  x\in  | \ a_j \cdot x \le b_j , \ \forall j \in [m]\} \subset \rn,
\end{align*}
where $a_j \in \mathbb R^n$ and $\|a_j\|_2 = 1$ for each $j$.
Additionally, we assume that $b_j\ge 1$
and also that $\Omega\subset B^2_R$, giving 
\begin{align*}
    B^2_1\subset \Omega\subset B^2_R.
\end{align*}

The task is to solve the following eigenvalue problem.
\begin{problem}[Drum]\label{prob:drum}
Given a polytope 
$    \Omega = \{  x\in \mathbb R^n | \ a_i \cdot x \le b_i , \ \forall i \in [m]\},$
\begin{align*}
    \text{minimize}& \qquad \lambda_0 \\ 
    \text{subject to}
        & \qquad \psi:\Omega \rightarrow \mathbb C, \\
        & \qquad \lambda_0 \psi(x) =  -\Delta\psi(x) \qquad \forall x \in \Omega \setminus \partial \Omega,\\
        & \qquad  \psi(x) = 0  \qquad \qquad \forall x \in \partial \Omega, \\
        & \qquad \psi \ne 0.
\end{align*}
\end{problem}

Similarly to the Schr\"odinger problem, we have an equivalent problem in terms of the energy functional (see Definition~\ref{def:energy}).
Here $\Delta^D_\Omega$ is the Laplacian operator on $\Omega$
with the Dirichlet boundary condition.
For the Drum problem, $\mathcal H = L^2(\Omega)$ is the appropriate Hilbert space.

\begingroup
\let\savedtheproblem\theproblem     
\addtocounter{problem}{-1}          
\renewcommand{\theproblem}{\savedtheproblem$'$} 
\begin{problem}[Drum]
Given a polytope 
$    \Omega = \{  x\in \mathbb R^n | \ a_i \cdot x \le b_i , \ \forall i \in [m]\},$
\begin{align*}
    \text{minimize} & \qquad  - \Delta^D_\Omega[ \psi] .
\end{align*}
\end{problem}
\endgroup

A simple application of the ASCOA gives the first known polynomial time algorithm for the drum problem.

\begin{theorem}
    \label{thm:DL_main}
Given $    \Omega = \{  x\in \mathbb R^n | \ a_i \cdot x \le b_i , \ \forall i \in [m]\}$ satisfying $ B^2_1\subset \Omega\subset B^2_R$,
Algorithm~\ref{alg:drum} solves Problem~\ref{prob:drum} up to error $\epsilon_0$ with probability $\ge2/3$ in polynomial time in $n,m,R,1/\epsilon_0$.
\end{theorem}

\begin{proof}
    See Section~\ref{sec:proof_of_main_thm}.
\end{proof}

\section{Algorithms}

Our algorithm applies an adiabatic evolution under the time‑dependent Hamiltonian
\begin{align}\label{eq:adiabatic_path}
H_Q(t) \;:=\; K_Q \;+\; V_Q(t), \qquad t\in[0,T],
\end{align}
acting on $n$ registers of $\log N$ qubits each. The kinetic term $K_Q$ is diagonal in the discrete Fourier basis and mimics the continuous operator $-\Delta$. 
The potential term $V_Q(t)$ is diagonal in the computational basis and mimics a continuous potential operator. 
This discretization of kinetic and potential operators follows prior work in quantum simulation~\cite{wiesner96, zalka96, KJLMA08, CLLLZ22}.

The time-dependent $V_Q(t)$ interpolates between an initial potential $V_Q(0)$ and a final potential $V_Q(T)$.
The initial potential is chosen so that $K_Q + V_Q(0)$ decomposes as a sum of single‑register Hamiltonians $H_N$. Consequently, the ground state of $H_Q(0)$ can be prepared in polynomial time as the product of the ground states of the $H_N$.
The final potential is defined so that $\lambda_0\big(H_Q(T)\big)$ is close to $\lambda_0(h)$, where $\lambda_0(A)$
is the minimum eigenvalue of a self-adjoint $A$.

We first describe how we discretize Schr\"odinger operators so that 
they are simulable on a quantum computer, and then 
state our two algorithms.

\subsection{Discretization of Schr\"odinger operators}
It is more natural to discretize a Schr\"odinger operator $h_\tn  = -\Delta_\tn + V_\tn$ on the torus 
\begin{align*}
    \tn:  = \rn / L\mathbb Z^n
\end{align*}
than $h$ on $\rn$.
Here, $\Delta_\tn = - \sum_{i\in[n]} \partial^2/\partial x_i^2$ is the Laplacian on the flat torus, and $V_\tn:\tn\rightarrow\mathbb R_{\ge 0}$ is smooth.
We discretize $\tn$ to the grid $\{\, yL/N \;|\; y\in\mathcal N^n \,\}$, where
    \begin{align}
        \mathcal N \;:=\; \left\{ -\frac{N}{2},\ldots,-1,0,1,\ldots,\frac{N}{2}-1 \right\}. \label{eq:grid_label}
    \end{align}    
We assume $N$ is a power of 2, and identify the labels for the computational basis states modulo $N$,
so that, for $y,y'\in \mathbb Z^n$,
\begin{align*}
    |y\rangle = |y'\rangle \qquad \text{whenever} \qquad y =  y' \mod N.
\end{align*}
This discretization requires $n\log  N$ qubits to represent the state.

We define the discretized Schr\"odinger, kinetic, and potential operators to be
\begin{align*}
    H_Q &:= K_{Q} + V_Q, \\
    K_Q &:= \sum_{i=1}^n I_N^{\otimes(i-1)} \otimes K_N \otimes I_N^{\otimes (n-i)}, \\
    V_Q &:= \sum_{y\in \mathcal N^n} |y\rangle \; V_{\tn}\!\left(\frac{yL}{N}\right)\! \langle y |, 
    \end{align*}
where $I_N$ is the $N\times N$ identity,
\begin{align*}
K_N \;:=\; \sum_{k_0 \in \mathcal N} U_N |k_0 \rangle \; \frac{4\pi^2 k_0^2}{L^2}\; \langle k_0 |\, U_N^\dagger ,
\end{align*}
and $U_N$ is the $N$‑dimensional quantum Fourier transform acting on $|k_0\rangle$ for $k_0\in \mathcal N$ by
\begin{align*}
    U_N|k_0\rangle \;=\; \frac{1}{\sqrt{N}} \sum_{y \in \mathcal N} e^{\,i 2\pi\, k_0 y/N}\, | y\rangle.
\end{align*}
Equivalently, we can also write 
\begin{align*}
    K_Q \;:=\; \sum_{k \in \mathcal N^n} U_{\mathcal F} |k \rangle \; \frac{4\pi^2 \|k\|_2^2}{L^2}\; \langle k |\, U_{\mathcal F}^\dagger
\end{align*}
where  $U_{\mathcal F} = \bigotimes_{i=1}^n U_N$ is the discrete Fourier transform on $\mathcal N^n$.

\subsection{Algorithm for Schr\"odinger convex optimization}
We define the time-dependent Hamiltonian 
\begin{align*}
    H_Q(t) \ := \ K_Q + \frac{T-t}{T}V_Q(0) + \frac{ t }{T }V_Q(T),
\end{align*}
where
\begin{align}
    V_{Q}(0) &:= b\sum_{ y\in\nn} |y\rangle \  \sum_{i\in[n]} \left( 1 - \Cut_{\frac{1}{4\sqrt n},\frac{1}{4\sqrt n}}\!\left(\left|\tfrac{y_i L}{N} \right| \right) \right)\langle y |,\label{eq:potential_def_i}\\
    V_{Q}(T) &:= \sum_{y \in\nn} |y\rangle \  \Sat_{c,1}\!\circ\!  V\big(\tfrac{yL}{N}\big)  \langle y| \label{eq:potential_def_f},
\end{align}
and $b,c \in\mathbb R_{\ge 0}$.
Here,
$\Cut_{\frac{1}{4\sqrt n },\frac{1}{4\sqrt n }}:\mathbb R\rightarrow \mathbb R$ is a 
smooth cutoff function and $ \Sat_{c,1}:\mathbb R\rightarrow\mathbb R$ 
is a smooth saturating function 
(Definition~\ref{def:smooth_functions}) 
such that
\begin{align*}
\Cut_{\frac{1}{4\sqrt n },\frac{1}{4\sqrt n }}(x)  
&
\
\begin{cases}
     = 1 \qquad\qquad & x \le \frac{1}{4 \sqrt n}, \\
     \in[0,1] \qquad\qquad &x\in [\frac{1}{4 \sqrt n}, \frac{1}{2 \sqrt n}],\\
     = 0 \qquad\qquad & x \ge \frac{1}{2 \sqrt n},
\end{cases}\\
\Sat_{c,1}(x)  
&
\
\begin{cases}
     = x \qquad & x \le c, \\
     :\text{monotone increasing} \qquad &x\in [c, c+1],\\
     : \text{some constant in} \  [c,c+1] \qquad & x \ge c+1.
\end{cases}
\end{align*}
Note that the initial Hamiltonian $H_Q(0)= K_Q + V_Q(0)$ is a sum of single-register Hamiltonians:
\begin{align}
    H_Q(0) = &\sum_{i=1}^n I_N^{\otimes(i-1)} \otimes H_N \otimes I_N^{\otimes (n-i)}, \nonumber\\
    \text{where}\qquad 
            H_N := &K_N \;+\; b\sum_{y_0\in \mathcal N} |y_0\rangle \left( 1 - \Cut_{\frac{1}{4\sqrt n}, \frac{1}{4\sqrt n}} \!\left( \left|\frac{y_0L}{N}\right| \right) \right)\! \langle y_0|, \label{eq:hn}
\end{align}

We present our main algorithm that we call the Adiabatic Sch\"odinger Convex Optimization Algorithm~(ASCOA) 
after the Fundamental Gap Theorem that we use crucially.
\begin{algorithm}[ASCOA]\label{alg:ASCOA}
Given an efficient circuit that computes $V(yL/N)$ for all $y\in\mathcal N^n$ up to a negligible error in $n,1/\epsilon$:
\begin{enumerate}
    \item Prepare the product ground state $|\Psi_{init}\rangle = \bigotimes_{i = 1}^n |\Psi_0\rangle$ of $H_Q(0)$ with probability $> 3/4$, by preparing $|\Psi_0\rangle$ on the $i$-th register for each $i\in[n]$ as follows:
    \begin{enumerate}
        \item Prepare $\Theta(N\log n)$ copies of the maximally mixed state $\frac{1}{N} I$ on $\log N$ qubits.
        \item Measure the energy of each copy with respect to $H_N$ (Equation~\eqref{eq:hn})
        using phase estimation.
        \item Move the copy with the lowest measured energy to the $i$‑th register.
    \end{enumerate}
    \item Evolve $|\Psi_{init}\rangle$ under the time‑dependent Hamiltonian
    \[
        H_Q(t) \;=\; K_Q +  (1- t/T)\, V_Q(0) \;+\; (t/T)\, V_Q(T), \qquad t\in[0,T],
    \]
    via Hamiltonian simulation to obtain the final state $|\Psi_{final}\rangle$, where $V_Q(0),V_Q(T)$ are as in Equations~\eqref{eq:potential_def_i}-\eqref{eq:potential_def_f}.
    \item Measure the energy of $|\Psi_{final}\rangle$ with respect to $H_Q(T)$ using phase estimation and output the result as an estimate of $\lambda_0(h)$.
\end{enumerate}
\end{algorithm}

Since we are only aiming for a polynomial-time algorithm, 
any off-the-shelf Hamiltonian simulation algorithm suffices.
For instance, we can use the algorithm from \cite{KSB18}.

\subsection{Algorithm for the drum problem}

Our approach for the drum problem is to run the ASCOA for a Schr\"odinger operator on $\rn$
with the potential defined by a barrier function 
that penalizes a function going outside of $\Omega$.
We use a smooth barrier function $\Bar_\epsilon$ (see Definition~\ref{def:smooth_functions}) such that
\begin{align*}
    \Bar_\epsilon(x)  \ 
    \begin{cases}
    = 0 \qquad& \qquad x \le 0 \\
    : \  \text{monotone increasing} \qquad & \qquad x \in[0,\epsilon] \\
    : \ \text{increasing with slope } 1 \qquad & \qquad x \ge \epsilon
    \end{cases}
    .
\end{align*}

\begin{algorithm} \label{alg:drum}
Given $    \Omega = \{  x\in \mathbb R^n | \ a_j \cdot x \le b_j , \ \forall j \in [m]\},
$
\begin{enumerate}
    \item Run Algorithm~\ref{alg:ASCOA} with  
\begin{align}
    V =\frac{3E}{\mu^6} \sum_{i\in[m]} \Bar_\epsilon( a_i \cdot x -b_i ) \label{eq:DL-potential}
\end{align}
    as an input potential with the parameters
    \begin{align*}
    &E = \Theta(n^2), \qquad \mu =   O({\epsilon_0}/{n^5  m^{1/3} R^{4}} )   . 
\qquad \epsilon = \Theta(\epsilon_0/n^2), \\
    &b=\Theta(n^{32}R^{24}/\epsilon_0^6)   ,\qquad  c=\Theta({ELm\sqrt n}/{\mu^6 })   , \qquad   L = 3R,  \qquad N = \text{poly}(n, m, 1/\epsilon, R )  .
    \end{align*}
    
\end{enumerate}  
\end{algorithm}
Note that the parameters are polynomials in $n,m,1/\epsilon,R$, 
giving polynomial runtime and number of required qubits.

\pagebreak

\section*{Part II: Proving Theorem~\ref{thm:main} and Theorem~\ref{thm:DL_main}}
\addcontentsline{toc}{section}{Part II: Proving Theorem~\ref{thm:main} and Theorem~\ref{thm:DL_main}}

\section{Preliminaries}\label{sec:prelim}

We adapt
the language and tools of operator theory and PDE theory.  
This section does not aim to provide a complete survey. We refer interested readers to \cite{RS75vol2, Hall13, Evans10}.

\subsection{Operator theory cheat sheet}
In this paper, we are concerned with functions $\psi:\Omega\rightarrow\mathbb C$
where $\Omega\in\{\rn ,B, \tn  \}$ and $B\subset\rn$ is a compact measurable subset.

\begin{itemize}
    \item A \textbf{Hilbert space} is a complete metric space under the metric induced by an inner product. 
For example, the space 
\begin{align*}
    L^2(\Omega) = \{ \ \psi \ | \ \psi :\mathbb R^n \rightarrow \mathbb C, \ \int_{\Omega } |\psi(x)|^2  \text d x < \infty  \}
\end{align*} is a Hilbert space, equipped with the inner product between $\psi,\phi\in L^2(\Omega)$ given by
\begin{align*}
    \langle \psi| \phi\rangle  :=\int_{ \Omega} \overline{\psi(x)} \phi(x) \ \text d x,
\end{align*}
where $\overline a$ is the complex conjugate of $a\in \mathbb C.$

    \item A \textbf{linear operator} $A$ on Hilbert space $\mathcal H$ is a linear map 
    that is defined by its \textbf{operator domain} $\mathcal D(A)\le \mathcal H$ 
    and its action on the operator domain $A:\mathcal D(A) \rightarrow \mathcal H.$
For example, the Laplacian operator 
\begin{align*}
   \Delta:\psi\rightarrow \sum_{i \in [n]}\frac{\partial^2 \psi}{\partial x_i^2}
\end{align*}
is an operator on $\mathcal H = L^2(\mathbb R^n)$ that is defined on the domain $\mathcal D({\Delta})$. 
Note that $\Delta\psi$ is 
only defined on twice-differentiable functions 
as per the definition above, therefore, $\mathcal D(\Delta) \subset L^2(\rn)\cap C^2(\rn)$.
Furthermore, not every twice-differentiable $\psi\in L^2(\rn)$
yields $\Delta \psi$ in $L^2(\rn).$
In many cases, we want an operator domain on which the operator is self-adjoint (see below).
Finding such a domain is a nontrivial task. 
For instance, see~\cite{RS75vol2}.

\item 
An important family of operators is \textbf{multiplication operators}. Given a function $f:\mathbb R^n \rightarrow \mathbb C$, multiplication operator $M_f$ is defined on $L_2(\mathbb R^n)$ so that for $x\in\mathbb R^n$,
\begin{align*}
    M_f\psi (x) := f(x) \psi(x).
\end{align*}
For example, the potential operator $M_V$
for a function $V:\rn\rightarrow \mathbb \mathbb R$ is a multiplication operator.
In this paper, we abuse the notation and denote $V$ instead of $M_V$,
following the notation in physics.

\item 
A \textbf{unitary operator} is a surjective linear operator $U:\mathcal H \rightarrow\mathcal H$ that preserves the inner product
\begin{align*}
    \langle\phi|\psi\rangle = \langle U\phi|U\psi\rangle \qquad \forall \phi,\psi\in\mathcal H.
\end{align*}

\item  
Define \textbf{operator norm} $\|A\|$ of an operator $M$ on $\mathcal H$ to be
\begin{align*}
    \|A\|:=\sup_{\psi \in \mathcal D(A)} \frac{\|A\psi\|}{\|\psi\|} = \sup_{\psi \in\mathcal D, \|\psi\| =1} \|A\psi\|.
\end{align*}
For a multiplication operator $M_V$, we have
$    \|M_f\| = \|f\|_{\infty}.$ 
\item An operator $A$ is called \textbf{symmetric} if its domain $\mathcal D (A)$ is dense in $\mathcal H$, and for all $\psi, \phi \in \mathcal D(A)$, we have $    \langle A\psi |\phi\rangle = \langle \psi |A\phi\rangle.$

\item 
Given an operator $A$ such that $\mathcal D(A)$
is dense in $\mathcal H$, its \textbf{adjoint operator} $A^\dagger$ is defined so that 
$A^\dagger \psi = z$ where $z$ is the unique vector satisfying $\langle  A\phi |\psi \rangle=\langle  \phi | z \rangle \ \forall \phi \in \mathcal D(A)$.
The domain $\mathcal D(A^\dagger)$ is the set of $\psi$ for which such $z$ exists.
\item 
The operator $A$ is \textbf{self-adjoint} if $\mathcal D(A) = \mathcal D(A^\dagger)$ and $A\psi = A^\dagger\psi $ for all $\psi \in \mathcal D(A)$.
\item 

A vector $\psi \in\mathcal D(A)\setminus \{0\}$ is an \textbf{eigenvector} of $A$ if there exists $\lambda \in \mathbb C$ such that $A\psi = \lambda \psi. $ The value $\lambda$ is called the \textbf{eigenvalue} of $\psi.$
\item 

The \textbf{spectrum} of $A$ is the set 
\begin{align*}
\text{Spec}(A) :=\{\lambda \ |A - \lambda I \text{ does not have an inverse with a bounded norm\}}.    
\end{align*}
If $\lambda$ is an eigenvalue, then $\lambda \in \text{Spec}(A)$. The inverse of this statement is not necessarily true. 

\item 
A spectrum $\text{Spec}(A)$ is \textbf{purely discrete} if each $\lambda\in \text{Spec}(A)$ is an eigenvalue, the multiplicity of $\lambda$, defined as $\dim \ker ( A - \lambda I)$, is finite, and $\text{Spec}(A)$ accumulates at no other point than $\infty$.

\end{itemize}

\subsection{Schr\"odinger operators}

A Schr\"odinger operator $h=-\Delta + V$ is a linear operator that is the sum of the negative Laplacian $-\Delta$,
and a multiplication operator $V$. 
Its action on $\psi\in \mathcal D(h)$ is defined by 
\begin{align*}
    h\psi(x)=  -\sum_{i\in[n]} \frac{\partial^2}{\partial x_i^2} \psi(x)  + V(x)\psi(x). 
\end{align*}
If $\Omega = B\subset\rn$ is a bounded domain,
we always impose the Dirichlet boundary condition, which asserts that an operator acts only on $\psi\in L^2(\Omega)$ such that 
\begin{align*}
    \psi(x)  = 0 \qquad \forall x\in\partial B.
\end{align*}
We use the superscript $D$ to denote that the operator is under the Dirichlet boundary condition.
For example, $\Delta_B^D$
is the Dirichlet Laplacian operator on the domain $B$.

In  this paper, we are interested in Schr\"odinger operators in the following forms.
\begin{definition}[Schr\"odinger operators]
    We define sets of Schr\"odinger operators on the real space $\mathbb R^n$, on a compact subset with smooth boundary $\Omega \subset \mathbb R^n$ with the Dirichlet boundary condition, 
    and on the length $L$ torus $\mathbb T^n :=\mathbb R^n/L\mathbb Z^n$ as follows:
    \begin{align*}
        S_{\mathbb R^n} :=& \{ h = -\Delta_{\mathbb R^n} + V \ |  \ V:\mathbb R^n \rightarrow \mathbb R_{\ge 0}, V\in \mathcal C^\infty(\mathbb R^n), \  V(x)\rightarrow\infty \text { as } |x| \rightarrow\infty \} \\
        S_{\Omega}^D :=& \{h= - \Delta_{\Omega}^D  +V \  | \  V: \Omega\rightarrow \mathbb R _{\ge 0}, \ V\in  C^\infty (\Omega\mathcal), \ \text{with the Dirichlet boundary condition}  \},   \\
        S_{\mathbb T^n} := &\{h= - \Delta_{\mathbb T^n}  +V \ |\  V: \mathbb T^n \rightarrow \mathbb R _{\ge 0}, \ V\in \mathcal C^\infty (\mathbb T^n) \},    
    \end{align*}
where $\Delta_{\mathbb R^n}$, $\Delta^D_{\Omega}$, $\Delta_{\mathbb T^n}$ denote the Laplacian operators $\sum_{i\in[n]}\partial^2/\partial x_i^2 $ on the respective domain with the respective boundary condition.
For each $h \in S_\rn \cup S^D_\Omega \cup S_\tn$, we assume that $h$ is self-adjoint on its domain $\mathcal D(h)$,
and $\mathcal D(h)$ contains $C^\infty_0(\Omega)$, the set of smooth functions with compact support in $\Omega$.
\end{definition}

\begin{remark}
    It is a non-trivial but well-known fact that for an operator of the form $h \in S_\rn \cup S^D_\Omega \cup S_\tn$, 
    we can define $\mathcal D(h)$ so that $h $ is self-adjoint
    and contains $C^\infty_0(\Omega)$. 
    See Theorem~X.28 of \cite{RS75vol2} for $\rn$, Theorem~3.8 of \cite{Krejcirik23} for a bounded $\Omega \subset \rn$, and the Sobolev space $H^2(\tn)$ forms an operator domain on which $\Delta_\tn$ is self-adjoint. 
\end{remark}

It is important to us that each $h \in S_\rn \cup S^D_\Omega \cup S_\tn$
has a purely discrete spectrum (Lemma~\ref{lem:completeness_of_eigenvecs})
since
the eigenfunctions of an operator with a purely discrete spectrum form a complete basis set of the Hilbert space (Lemma~\ref{lem:completeness}).
Furthermore, an eigenfunction of a Schr\"odinger operator is smooth as a function (Lemma~\ref{lem:smoothness_of_eigenvecs}).

\begin{lemma}[Theorem 3.10 of~\cite{Krejcirik23}]\label{lem:completeness}
Let $A$ be a self-adjoint operator with a purely discrete spectrum on the Hilbert space $\mathcal H$. Then, the eigenvectors of $A$ form an orthonormal basis in $\mathcal H$.
\end{lemma}

\begin{lemma}[Spectrum of a Schr\"odinger operator]\label{lem:completeness_of_eigenvecs}
Suppose $h \in S_\rn \cup S^D_\Omega \cup S_\tn$. Then, $\text{Spec}(h)$ is purely discrete, and the eigenvectors form a complete orthogonal basis.
\end{lemma}
\begin{proof}
By Lemma~\ref{lem:completeness}, it is enough to show that $h$ has a purely discrete spectrum.
   For $h \in S_\rn$, see \cite[Theorem~XIII.16]{RS78_vol4}. For $h \in S^D_\Omega \cup S_\tn$, Weyl's essential spectrum theorem applied to
   the corresponding Laplacian operators gives the purely discreteness of the spectrum.
\end{proof}

\begin{lemma}[Smoothness of eigenfunctions]\label{lem:smoothness_of_eigenvecs}
    Let $h\in S_\rn\cup S^D_\Omega \cup S_\tn $. Let $\psi \in \mathcal D(h)$ be an eigenfunction of $h$. Then $\psi$ is a smooth function.
\end{lemma}
\begin{proof}
An eigenfunction of $h$ is a solution of an elliptic PDE, whose smoothness is a well-studied problem.
The hypoellipticity of Schr\"odinger operators on $\rn$~\cite{hormadner61},
the elliptic regularity on a compact, smooth-boundary domain with the Dirichlet boundary~\cite[Theorem 6.3.6]{Evans10} and the elliptic regularity on $\tn$~\cite{cass16}
give the smoothness of the eigenfunction of $h$ in each case.

\end{proof}

\subsection{Smooth functions}

\begin{definition}[useful smooth functions]\label{def:smooth_functions}

   We define a smooth bump function, a smooth cutoff function, a smooth saturating function,
   and a smooth barrier function
\begin{align*}
   \text{Bump}(x)
    \ &  = \ 
   \begin{cases}
      \ \exp \left(\frac{1}{x(x-1)}\right) \qquad &x \in[0,1],
\\ \ 0 \qquad &x \notin [0,1],
   \end{cases} 
   \\
    \text{Cut}_{\alpha,\beta} (x) 
\ & =   \ 1 - {\int_\alpha^x \Bump \big(\frac{y - \alpha}{\beta } \big) \ \text d y} \ \Big/ {\int_\alpha^{\alpha +  \beta} \Bump \big(\frac{y - \alpha}{\beta } \big) \ \text  d y } ,\\
    \Sat_{\alpha, \beta}(x) 
    \  &= \ \int_0^x \Cut_{\alpha,\beta}(y) \ \text d y,   
    \\
     \Bar_\epsilon(x)
     \ & = \ \int_0^x(1 - \Cut_{0,\epsilon}(y)) \ \text d y.
\end{align*}

\end{definition}

\section{Proof outline and techniques}\label{sec:overview}
The analysis achieves two main goals.
First, we want an inverse-polynomial spectral gap in $H_Q(t)$ throughout $t\in[0,T]$
so that we can apply the Adiabatic Theorem (Theorem~\ref{thm:adiabatic}).
Second, we want the lowest eigenvalue of $h$
to be close to that of the final qubit Hamiltonian $H_Q(T)$.
Once we have established these points, the correctness of the algorithm follows from the Adiabatic Theorem.

In general, proving a lower bound on the spectral gap is challenging.
In our case, we leverage the Fundamental Gap Theorem, which ensures an inverse polynomial gap for
a Dirichlet Sch\"rodinger operator with a convex potential on a bounded convex domain.

We consider the four pairs of a geometric domain and a Hamiltonian 
\[
(\rn,h),\quad  (B, h_B^D),\quad  (\tn,h_{\tn}),\quad (\mathcal N^n,H_Q),
\]
where $B$ is a Euclidean ball in $\rn$. Each plays a different role in the analysis:
$h$ on $\rn$ is the objective,
$h^D_B$ on $B$ provides the spectral gap via the Fundamental Gap Theorem,
$h_\tn$ on $\tn$ plays the role of a ``gearbox'' that connects other Hamiltonians,
and finally $H_Q$ on the finite set $\mathcal N^n$ is efficiently simulable on a quantum computer.
We connect them through truncation lemmas (Section~\ref{sec:truncation}).

\subsection{The Adiabatic and the Fundamental Gap Theorems}
The central facts in proving the correctness of our algorithm are the Adiabatic Theorem
and the Fundamental Gap Theorem.
For the rest of the paper, $\lambda_0(H), \lambda_1(H)$ denote the lowest and the second lowest eigenvalues of $H$,
and
\begin{align*}
    \gap(H):=\lambda_1(H)  - \lambda_0(H).
\end{align*}

\begin{theorem}[The Adiabatic Theorem (adapted from~\cite{rei04, AvD+04})]\label{thm:adiabatic}
Let $H_{ init}$ and $H_{final}$ be two Hamiltonians acting on a finite-dimensional quantum system. 
Consider the time-dependent Hamiltonian
\[
  H(t) := (1 - t/T) H_{ init} + (t/T) H_{final},
\]
that has a unique ground state for all $t\in [0,T]$. 

Then the final state of an adiabatic evolution according to $H(t)$ for $t\in[0,T]$
is $\epsilon$-close to the ground state of $H_{final}$,
if the total evolution time 
\[
  T \ge \Omega\!\left(
    \frac{\|H_{\rm final} - H_{\rm init}\|^{2}}{\epsilon (\min_{t \in [0,T]} \gap (H(s)))^3}
  \right).
\]  
The operator norm is the spectral norm $\|H\| := \max_{\|w\|=1} \|H w\|$.
\end{theorem}
We apply this theorem to the time-dependent qubit theorem $H_Q(t)$
of Algorithm~\ref{alg:ASCOA}.

The Dirichlet Schr\"odinger operator $h^D_B$ is
a Schr\"odinger operator on $B$ with the Dirichlet boundary.
The following theorem~\cite{AC10} provides a lower bound on the spectral gap of $h^D_B$ as an inverse polynomial in the diameter of $B$.
\begin{theorem}[Fundamental Gap Theorem~\cite{AC10}] \label{thm:fundamental_gap}
    Let $\Omega \subset \mathbb R^n$ be a bounded convex domain of diameter $R$, and $V$ a convex potential. 
    Then the eigenvalues of the Dirichlet Schr\"odinger
operator $h^D_\Omega = -\Delta^D_\Omega + V$ satisfy
\begin{align*}
   \gap(h^D_\Omega) \ge \frac{3\pi^2}{R^2},
\end{align*}
where $-\Delta^D_\Omega$ denotes the Dirichlet Laplacian operator on $\Omega$.
\end{theorem}

\subsection{Low-energy truncation}
A disconnection between Theorem~\ref{thm:adiabatic} and Theorem~\ref{thm:fundamental_gap} is that
we need a bound on $\gap(H_Q(t))$,
whereas we have a bound on $\gap(h^D_B)$.

The key idea
for resolving this issue
is that 
to bound the spectral gap,  
we only need to care about the low‑energy properties; 
the rest of the Hilbert space can be ignored, even if infinite‑dimensional.
To formalize this, we introduce $(E,\epsilon)$‑truncated domains that capture the low‑energy subspace up to an error $\epsilon$.

We then prove the Truncation Lemmas showing that 
if 
$(E,\epsilon)$‑truncated domains of two Hamiltonians
admit an isomorphism between them 
that approximately preserves energy, 
then their first two eigenvalues must be close to each other.
Our truncation framework allows a highly modularized analysis
of the spectral gap.

\begin{definition}[Truncation and Equivalence]\label{def:truncation}
For $\epsilon>0$,
a normalized vector $\widetilde{\psi}\in\mathcal D(h)$ is an \textbf{$\epsilon$-truncation} of $\psi$ with respect to $h$, 
if
\begin{align*}
   i)  &  \qquad  h[\widetilde{\psi}] \  \leq \   h[\psi] + \epsilon \qquad \qquad \text{(energy)} \\  
  ii)  & \qquad \|\widetilde{\psi} - \psi\| \  \leq \ \epsilon \qquad \qquad \text{(norm)}.
\end{align*}

A subspace \( \widetilde{\mathcal{D}} \le \mathcal H\) is an \textbf{\((E, \epsilon)\)-truncated domain} of $h$, 
if, for every normalized $\psi \in \mathcal{D}(h)$ such that $h[\psi]\le E$, 
there exists an $\epsilon$-truncation of $\psi$ in $ \widetilde{ \mathcal D}$.

Self-adjoint operators $h_a$ and $h_b$ are \textbf{$(E,\epsilon)$-equivalent}
if there exist $(E,\epsilon)$-truncated domains $\widetilde{\mathcal D}_a $ of $h_a$ and $\widetilde{\mathcal D}_b $ of $h_b$
that are isomorphic via a unitary
$U: \widetilde{\mathcal D}_a \rightarrow \widetilde{\mathcal D}_b$
such that 
\begin{align*}   
            \left| h_a[\widetilde{\psi}]  - h_b[U\widetilde{\psi}] \right|     \leq  \epsilon \qquad \qquad \forall  \ \widetilde{\psi} \in \widetilde{\mathcal D}_a .        
    \end{align*}
\end{definition}

\begin{lemma}[$\lambda_0,\lambda_1$ approximation]\label{lem:truncation_combined}
Let $h_a,h_b\ge 0$ be a self-adjoint linear operator with purely discrete spectrum. 
Suppose $h_a$ and $h_b$ are $(E,\sigma)$-equivalent and the following conditions hold:
\begin{align}
&\gap(h_a)\ge g \ > 0, \nonumber \\
&2({\lambda_1(h_a)} +1) \  \leq \  E,\quad \nonumber\\
&\epsilon, E^{-1}  \ < \ 1 . \nonumber
\end{align}
If $\epsilon \in [0,c]$ for a sufficiently small universal constant $c$
and
\begin{align*}
    \sigma =  O \left(\frac{\epsilon g}{E^{1.5}} \right),
\end{align*} then
\begin{align*}
|\lambda_0(h_a) - \lambda_0(h_b)|, \ 
|\lambda_1(h_a) - \lambda_1(h_b)| \leq \epsilon . 
\end{align*}

\end{lemma}
\begin{proof}
    See Appendix~\ref{sec:truncation}. The proof is by elementary linear algebra.
\end{proof}

\subsection{Truncation in position}
We explain how we apply the idea of truncation to Schr\"odinger operators of our interest.
We assume that $V$ is of a bowl shape: low around the center (the condition $B^2_1\subset V^{-1}_a$), 
and high outside the radius $r$ (the condition $V^{-1}_b \subset B^2_{r}$).
This assumption enables Markov's inequality, which is employed to show that a wavefunction with a low energy has a low weight outside the ball $B^2_{r}$,
due to the high potential.
We apply this position-based Markov's inequality to truncate the Schr\"odinger operators on $\rn$, $B:=B^2_{r+1}$, and $\tn$.
These truncations in fact $(E,\epsilon)$-equivalent, showing that the first two eigenvalues of the three Hamiltonians
close to each other.

\begin{lemma}\label{lem:pos_based_equivalences}
Suppose $h_1, h_2, h_3$ are Schr\"odinger operators on $\rn, B :=B^2_{r+1}, \tn$, respectively, and $h_2$ is under the Dirichlet boundary condition
\begin{align*}
h_1&:=-\Delta_{\mathbb R^n} + V_1 \in S_{\mathbb R^n}, \\
h_2&:=- \Delta_B^D + V_2 \in S_{B}^D, \\
h_3 &:= - \Delta_{\mathbb T^n} + V_3 \in S_{\mathbb T^n} , 
\end{align*}
where $2(r+1)< L$.
Furthermore, assume that $\epsilon<1$ and
\begin{align*}
\begin{cases}
    V_1(x)=V_2(x)=V_1(x) \ \ \ \ &\text{if} \ \ \|x\|_2\le r+1, \\
    V_1(x),V_2(x),V_3(x) \ge b \ \ \ \ &\text{if} \ \ \|x\|_2\ge r.
\end{cases}    
\end{align*}

If $b= \Omega( E^7/\epsilon^6)$, then $h_i$ and $h_j$ are $(E,\epsilon)$-equivalent for any $i,j\in \{1,2,3\}.$
\end{lemma}

\begin{proof}
    See Section~\ref{sec:truncation_in_position_appendix}.
    The proof is by elementary calculus.
\end{proof}

\subsection{Truncation in frequency and discretization error}
We can similarly apply Markov's inequality in the frequency domain.
The kinetic term on the torus penalizes high-frequency components. 
Therefore, 
a low-energy state should have
low weights
on the high frequency components.  
Hence, the low-frequency spaces are natural truncated domains for the torus and the qubit Schr\"doinger operators.

The discretization error arises in this step, where we upper bound it by a quantifiable notion of smoothness $\log_\partial$, that we explain shortly in the next subsection.
Intuitively, 
the following lemma states that 
we need more qubits per spatial dimension 
as $L$ gets bigger, $V_\tn$ gets less smooth, and the target error gets smaller.

\begin{lemma}[Equivalence of torus and qubit Hamiltonians]\label{lem:torus_qubit_equiv}
   Let $V_\tn:\mathbb T^n \rightarrow \mathbb R$ be a smooth potential on the torus, and
   \begin{align*}
       V_Q:=\sum_{y \in \mathcal N^n}|y\rangle V_\tn \big(\frac{yL}{N}\big)\langle y |.
   \end{align*}
Suppose 
\begin{align*}
h_{\mathbb T^n} & : = -\Delta_{\mathbb T^n} +V_\tn \qquad \qquad  \qquad   \text{on \ }\mathbb T^n, \\
h_{Q}  & : = K_{Q} +V_{Q}  \qquad \qquad \qquad \qquad \text{on \ }n\times\log N \ \text{qubits,}
\end{align*}
where $V_\tn(\mathbb T^n)\subset [0, V_{\max}] $ and $\log_\partial (V_\tn, \mathbb T^n, m) \le \log p(m)$ for some polynomial $p$.

    For $\epsilon<1$, we have that $h_{\mathbb T^n}$ and $h_Q$ are $(E,\epsilon)$-equivalent, if
    \begin{align*}      N\ge \Omega\left( \frac{L^2  (p(2n))^2  (E+V_{\max})^{3/2} }{\epsilon }\right).
    \end{align*}
\end{lemma}
\begin{proof}
    See Section~\ref{sec:truncation_in_frequency_appendix}. 
\end{proof}

\subsection{Smoothness factor}
By considering the discrepancy between a Fourier base state on $\tn$ and its discretization on $\mathcal N ^n$,
we naturally arrive at the following definition of a quantifiable smoothness.

\begin{definition}
    Given a smooth function $g:A \rightarrow \mathbb R$ on a subset $A$ of $\mathbb R^n$ or $\mathbb T^n$, we define the \textbf{smootheness factor} of $g$ to be
    \begin{align*}
       \log_\partial (g,A,m) =\max_{l\in[m], j\in[n], x\in A} \frac{\left( \log |\partial^{l}_j g(x)| \right)^+}{l} ,
    \end{align*}
where we use the notation $(a)^+ = \max(a,0)$ for $a \in \mathbb R$, and assume that the support of $\partial_j g$ 
is compact for all $j$. 

\end{definition}
Intuitively, the smoothness factor indicates how ``rough" a function is. 
In this paper, it is desired that $V_\tn$ has a logarithmic smoothness factor, so that there exists some polynomial $p(m)$ such that 
\begin{align*}
\log_\partial (g,A,m) \le \log p(m),    
\end{align*}
as needed to apply Lemma~\ref{lem:torus_qubit_equiv}.

Having a logarithmic smoothness factor is closed under summation, scalar multiplication, and most importantly, composition.

\begin{lemma}
The following statements are true: 
\begin{enumerate}
       \item  \emph{(Scalar multiplication and summation)}   Let $f,f_k:\mathbb  T^n \rightarrow\mathbb R$ be smooth for $k\in[r]$, and $c>0$. 
    Then we have
    \begin{align*}
       \log_\partial( cf, \mathbb T^n, m)
      &\  \le \   (\log c)^+ \ + \     \log_\partial( f, \mathbb T^n, m) ,
      \\
       \log_\partial(\sum_{i\in[r]} f_i, \mathbb T^n, m)
      &\  \le \   \log r \ + \  \sum_{i\in[r]}     \log_\partial( f_i, \mathbb T^n, m) .
    \end{align*}
    \item   \emph{(Composition)}  Let $g:\mathbb  T^n \rightarrow\mathbb R$ and $f:g(\mathbb T^n) \rightarrow \mathbb R$  be smooth, where $g(\mathbb T^n)$ is the range of $g$. 
    Then we have
    \begin{align*}
       \log_\partial( f\circ g, \mathbb T^n, m)
      \  \le \   2 \log m \ + \  \log_\partial(f, g(\mathbb T^n), m)   \ +  \ 
       \log_\partial  (g, \mathbb T^n, m). 
    \end{align*}
    \item \emph{(The four smooth functions)} We have
    \begin{align*}
    \log_\partial (\Bump, \mathbb R,m) &\le \log O(m^4), \\    
     \log_\partial (\Sat_{\alpha,\beta}, \mathbb R,m), \log_\partial (\Cut_{\alpha,\beta}, \mathbb R,m),\log_\partial (\Bar_{\beta}, \mathbb R,m)     &\le  \log O(1/\beta + m^4). \end{align*}
\end{enumerate}
 \end{lemma}
\begin{proof}
    See Section~\ref{sec:smoothness}. The proof is by elementary calculus.
\end{proof}

\subsection{The drum problem}

Our approach is to run ASCOA for a Schr\"odinger operator on $\rn$
with the potential defined by a barrier function that penalizes going outside of $\Omega$, namely
\begin{align*}
    V =\frac{3E}{\mu^6} \sum_{i\in[m]} \Bar_\epsilon( a_i \cdot x -b_i ).
\end{align*}

To analyze, 
we define a slightly expanded region of $\Omega$,
\begin{align*}
    \Omega'  := (1+\epsilon)\Omega  =  \{  x\in \mathbb R^n | \ a_i \cdot x \le b_i(1+\epsilon) , \ \forall i \in [m]\},
\end{align*}
and two associated Schr\"odinger operators
\begin{align} \label{eq:DL_hamiltonians}
      h &:=-\Delta + V \\
    h^D &:=-\Delta^D_{\Omega' } + V|_{\Omega'}, 
  \end{align}
  where $h$ is a Schr\"odinger operator on $\rn$, and $h^D$ is its Dirichlet restriction on $\Omega'$ .

We show that $|\lambda_0(-\Delta^D_\Omega)- \lambda_0(h)|$ 
is small, and hence we only need to compute $\lambda_0(h)$
by running the ASCOA.
We first show 
\begin{align*}
|\lambda_0(h^D) - \lambda_0(-\Delta^D_\Omega)| \le O(\epsilon_0)    
\end{align*}
by using basic facts about Schr\"odinger operators
(Lemma~\ref{lem:drum_hD-DeltaD}).
Then, we show that
\begin{align*}
|\lambda_0(h) - \lambda_0(h^D)| \le O(\epsilon_0)    
\end{align*}
by showing that $h$ and $h^D$ are $O(\epsilon_0)$-equivalent
if $\mu$ is polynomially large
(Lemma~\ref{lem:DL_equivalence}).

Therefore, it is enough to compute $\lambda_0(h)$ through Algorithm~\ref{alg:ASCOA}. 
The parameters in the algorithm for $V$ are polynomial in $n,m,1/\epsilon,R$.

\section{Proofs of Theorem~\ref{thm:main} and Theorem~\ref{thm:DL_main} }\label{sec:proof_of_main_thm}
We first provide some facts on the eigenvalues of the Schr\"odinger operators and prove the two theorems.
\subsection{Monotonicity relations on eigenvalues}
The following lemma characterizes the $i$-th eigenvalue of $h$.

\begin{theorem}[Min-Max] \cite[Theorem 4.10]{teschl09}
     Let $h$ be self-adjoint with a purely discrete spectrum. 
     Let $\lambda_0(h)\le \lambda_1(h)\le \cdots$ 
     be the eigenvalues of $h$.
     
Then, we have
\begin{align*}
    \lambda_n(h) \ =\  \min_{\psi_0,\dots,\psi_{n-1}} 
    \max \{\ \langle \psi|h\psi\rangle \ |\ \|\psi\|=1, \psi\perp\psi_i \ \forall i\in\{0,\dots, n-1\}\}.
    \end{align*}
\end{theorem}
As a corollary, we have a comparison of each eigenvalue of two Schr\"odinger operators, 
if their potentials are comparable~\cite[Corollary~4.13]{teschl09}. 
\begin{lemma}[Potential comparison]\label{lem:potential_comparison}
Let $h_1,h_2 \in S^D_\Omega$ with potentials $V_1,V_2:\Omega\rightarrow \mathbb R_{\ge0}$ such that
$V_1(x)\ge V_2(x)$ for all $x\in\Omega$.
Then, $\lambda_k(h_1)\ge \lambda_k(h_2)$ for all $k\in\mathbb Z_{\ge 0}$.    
\end{lemma}

The following Lemma is a standard result in the spectral theory~\cite{Welsh1972}
that says, for any $k$, the $k$-th eigenvalue decreases as the domain increases. 
Intuitively, one could view the Dirichlet Schr\"odinger operators on $\Omega$ as
a Schr\"odinger operator on $\rn$ with a potential $V$ such that 
$V(x)=\infty$ for $x\notin \Omega$.
\begin{lemma}[Domain monotonicity]\label{lem:domain_monotonicity}
Suppose for bounded $\Omega_1\subseteq\Omega_2 \subset \rn$,
\begin{align*}
h &:=-\Delta_{\rn} + V \in S_{\rn},
\\
h_1 &:=-\Delta^D_{\Omega_1} +V\big|_{\Omega_1} \in S^D_{\Omega_1},
\\
h_2 &:=-\Delta^D_{\Omega_2} +V\big|_{\Omega_2} \in S^D_{\Omega_2}.
\end{align*}
Then, for any $k\in\mathbb Z_{\ge 0} $, we have 
\begin{align*}
    \lambda_k(h)\le \lambda_k(h_1) \le \lambda_k(h_2).
\end{align*}
If $V =0$, then we have the special case of 
the domain monotonicity on the Dirichlet Laplacian:
\begin{align*}
    \lambda_k(-\Delta^D_{\Omega_1}) \le
    \lambda_k(-\Delta^D_{\Omega_2}) . 
\end{align*}
 \end{lemma}

The following lemma is what we want to use in the analysis of the main theorem.
\begin{lemma}\label{lem:lambda1_upperbound}
    Let $h =t h_1 +(1-t)h_2$ for $t\in[0,1]$ be a convex combination 
    of two Dirichlet Schr\"odinger operators 
    $h_1,h_2$ on $\Omega\subset\mathbb  R^n$ 
    with smooth potentials $V_1,V_2:\Omega \rightarrow\mathbb R_{\ge0}$.
    Suppose $V_2(x)  =0$ for all $x\in \Omega'\subset \Omega$, where $\Omega'$ is compact. 

    Then, 
    \begin{align*}
        \lambda_1(h) \le \lambda_1(-\Delta^D_{\Omega'}) +\max_{x\in\Omega'} V_1(x).
    \end{align*}
\end{lemma}
\begin{proof}
We have $h = -\Delta^D_\Omega + tV_1+(1-t)V_2$.
    By the domain monotonicity (Lemma~\ref{lem:domain_monotonicity}), we have $\lambda_1(h)\le \lambda_1(h^D_{\Omega'}) $, 
    where $h^D_{\Omega'}:=-\Delta^D_{\Omega'} +V_1\big|_{\Omega'} $
    is a Dirichlet Schr\"odinger operator 
    on $\Omega'$ with the potential given as the restriction of $V_1$.

    By Lemma~\ref{lem:potential_comparison}, we have
    \begin{align*}
     \lambda_1(h^D_{\Omega'}) \le \lambda_1(-\Delta^D_{\Omega'}  +\max_{x\in \Omega'} V_1(x)) =\lambda_1(-\Delta^D_{\Omega'}) +\max_{x\in \Omega'} V_1(x),  
    \end{align*}
    since $V_1(y)\le \max_{x\in \Omega'}V(x)$ at all $y\in\Omega'$.
    Therefore, we have 
    \begin{align*}
        \lambda_1(h)\le \lambda_1(h^D_{\Omega'}) \le \lambda_1(-\Delta^D_{\Omega'}) +\max_{x\in\Omega'} V_1(x).
    \end{align*}
     
\end{proof}

\subsection{Proof of Theorem~\ref{thm:main}}
\textbf{Theorem 1} \ (restate).      
    Algorithm~\ref{alg:ASCOA} computes the lowest eigenvalue of $h = -\Delta + V$ within error $\epsilon_0$ with probability $\ge 2/3$ in polynomial time in $n, 1/\epsilon_0,L,c$,
    if $V$
    satisfies the 
    following conditions, where 
    $B^q_d:=\{x\in\rn \ | \  \|x\|_q\le d\}$ and 
     $V_E^{-1}:= V^{-1}([0,E])$:
    \begin{enumerate}
        \item  \emph{Purely discrete spectrum and convexity:} we have $V(x)\rightarrow\infty$ as $\|x\|\rightarrow\infty$ and $V$ is convex on $B^2_{r+1}$.
        \item  \emph{``Bowl-shaped'' in $B^\infty_{L/2}$:} there exist $a\le b\le c$, and $1 \le r < L/2$ such that
        \begin{align*}
            B^2_1\subset V^{-1}_{a } \subset V^{-1}_{b}\subset B^2_r \subset B^2_{r+1} \subset V^{-1}_{c}\subset V^{-1}_{c+1} \subset B^\infty_{L/2}
        \end{align*}  
         where
        \begin{enumerate}[(i)]
            \item $b = \Theta(E_0^7/\sigma^6)$
            \item $E_0=10(n(n+3)\pi^2 +a) = \Theta(n^2+a)$
            \item $\sigma = \Theta(\epsilon_0/r^4 E_0^{1.5})$
            \item $\epsilon_0<1 < r,E_0,c$. 
        \end{enumerate}
        \item \emph{Bounded high-order derivatives:} the potential $V$ is smooth and 
        \begin{align*}
            \max \left\{ { \frac{1}{l}\log|\partial^l_j V(x)| }  \  \bigg|  \ {l\in[m], j\in[n], x\in [-L/2,L/2]^n }\right\}  \le \log p(m)
        \end{align*} for some polynomial $p$.
        \item \emph{Access to $V$:} there exists a circuit that computes $V(x)$ for all $ x \in B^\infty_{L/2} $ with a negligible error in $n, 1/\epsilon_0,  L, c$ in polynomial time in the same parameters.
    \end{enumerate}

\begin{proof}[Proof of Theorem~\ref{thm:main}]
We define functions $V_{init},V_{final}, W_t:B^\infty_{L/2} \rightarrow\mathbb R_{\ge 0}$ for $t \in [0,T]$ to be
\begin{align*}
    V_{init}(x) &:=       b   \sum_{i\in[n]} \left( 1 - \Cut_{\frac{1}{4\sqrt n},\frac{1}{4\sqrt n}}\!\left(\left| x_i \right| \right) \right)     ,  \\
    V_{final}(x) &:=   \Sat_{c,1}\!\circ\!  V(x),  \\
    W_t(x) &: =  \left(1- \frac{t}{T}\right) V_{init}(x) +  \frac{t}{T}V_{final} (x).  
\end{align*}

Also define time-dependent Hamiltonians on $B:=B^2_{r+1}, \mathbb T^n :=\mathbb R^n/L\mathbb Z^n$,
and $\mathcal N^n$
\begin{align*}
    H_{B}(t)  &: = -\Delta^D_B + W_{B,t} , \\
    H_{\mathbb T^n}(t)  &: = -\Delta_{\mathbb T^n} + W_{\mathbb T^n,t}, \\
    H_{Q}(t) &:= \  \  K_Q  + \sum_{y \in \mathcal N} |y\rangle W_t\big({yL}/{N}\big)  \langle y |,   
\end{align*}
where $W_{B,t}:= W_t|_B$ is the restriction of $W_t$ on $B$,
and
$ W_{\mathbb T^n,t}(x) := W_t(x')$ for $x '\in[-L/2,L/2]^n$, such that $  x =x'\mod L\mathbb Z^n$. 
Note that $H_Q(t)$ is defined consistently as in Algorithm~\ref{alg:ASCOA}.

We use the notation $g:={3\pi^2}/{(2(r+1))^2}$ for the rest of the proof.

\paragraph{Step 1:} show $\gap(H_{B}(t))\ge g$ and $2(\lambda_1(H_B(t))+1)\le E_0$ for all $t\in[0,T]$.  
At an arbitrary time $t \in [0,T]$,
we know from Theorem \ref{thm:fundamental_gap} that
\begin{align*} 
    \gap(H_B(t)) 
    \ge \  \frac{3\pi^2}{(2(r+1))^2}=g, 
\end{align*}
since the diameter of $B$ is $2(r+1).$

To upper bound $\lambda_1(H_B(t))$, we use the fact that $B$ includes the box $B':=B^\infty_{1/2\sqrt{n}}\subset B^2_1$,
where the potential is at most $a$.
We define
\begin{align*}
H_{B'}(t):= -\Delta^D_{B'} + W_{t}|_{B'}    .
\end{align*}

By the domain monotonicity (Lemma~\ref{lem:domain_monotonicity}), we have 
$\lambda_1(H_B(t))\le \lambda_1(H_{B'}(t))$.
We have $ W_{t}|_{B'}(x)  \le a$ for $x \in B'$, by the condition $B^2_1 \subset V^{-1}_{a}$.
Therefore, by Lemma~\ref{lem:lambda1_upperbound},
\begin{align*}
    \lambda_1(H_B (t)) 
    &\le  \lambda_1(H_{B'}(t)) \\
    &\le \lambda_1(-\Delta^D_{B'}) + a \\
    &=  n(n+3) \pi^2 + a\\
    &= E_0/10. 
\end{align*}

\paragraph{Step 2:} show $\gap(H_{\mathbb T^n}(t))\ge 0.99g$, and $2(\lambda_1(H_{\mathbb T^n}(t))+1)\le E_0$ for all $t\in[0,T]$. 
We first show that $H_B(t)$ and $H_{\tn}(t)$ are $(E_0,\sigma)$-equivalent, 
and then, transfer the large gap from $B$ to $\mathbb T^n$ via Lemma~\ref{lem:gap_approx}.

For all $x$ with $\|x\|_2\ge r$, 
we have $V_{init}(x) >b$
and, due to the condition $V^{-1}_b\subset B^2_r $, we have
\begin{align*}
    V_{final}(x) = \Sat(V(x))\ge \min(V(x), c) > b,
\end{align*}
where we denote $\Sat: = \Sat_{c,1}$.
Therefore, we have $W_t(x)\ge \min(V_{init}(x),V_{final}(x))> b $.
By Lemma~\ref{lem:pos_based_equivalences}, it follows that $H_{B}(t)$ and $H_{\mathbb T^n}(t)$ are $(E_0,\sigma)$-equivalent.

In addition, by Lemma~\ref{lem:second_approx}, we have that $\lambda_1(H_{\mathbb T^n}(t))  \le E_0/10 + O(\epsilon_0)  $.
Therefore, we have $2(\lambda_1(H_{\mathbb T^n}(t) +1)\le E_0$, and Lemma~\ref{lem:gap_approx} is applicable.

Since $\sigma = \Theta(\epsilon_0/r^4 E_0^{1.5}) =   O(g^2/E_0^{1.5}) $, 
by Lemma~\ref{lem:gap_approx}, we have 
\begin{align*}
    \gap(H_{\mathbb T^n}(t)) \ge \gap(H_{B}(t)) - 0.01g \ge  0.99g. \label{eq:main_torus_gap}
\end{align*}

\paragraph{Step 3:} show that $\gap (H_Q(t))\ge0.98g$ for all $t\in[0,T]$.
Again, we show that $H_{\tn}(t)$ and $H_{Q}(t)$ are $(E_0,\sigma)$-equivalent
and then apply Lemma~\ref{lem:gap_approx} to lower bound the gap.

To argue equivalence between $H_{\mathbb T^n }(t) $ and $H_{Q}(t)  $ via Lemma~\ref{lem:torus_qubit_equiv}, 
we need an upper bound on $\log_\partial(W_{\mathbb T^n,t},\mathbb T^n, m)$.
We first show that the function $W_{\mathbb T^n,t} $ is smooth at all $t\in[0,T]$, which amounts to 
showing smoothness of $V_{final}$ and $V_{init}$.
For $x$ with $\|x\|_\infty < L/2 $, $\Sat\circ V$ is smooth because it is a composition of two smooth functions.
For $x$ with $\|x\|_\infty = L/2 $, we have $V(x) \ge c +1 $ by 
the condition $V^{-1}_{c+1}\subset B^\infty_{L/2}$;
hence $\Sat(V(x))$ is constant, and all of its derivatives vanish. 
Also, $V_{init} $ is smooth, because each $\Cut (|x_i|) $ is smooth on the torus, where we denote $\Cut :=\Cut_{1/4\sqrt n,1/4\sqrt n}$.
Therefore, $W_{\mathbb T^n,t} $ is smooth on the torus at all $t$.

By Condition 2, we have $V_{final} (\mathbb T^n)\subset [ 0, c +1] $, and 
we have $V_{init}(\mathbb T^n) \subset [0,n b]  $ by definition.
Therefore, $W_{\mathbb T^n,t}(\mathbb T^n) \subset [0,O(c + nb )]. $

We bound the smoothness factor of $W_{\mathbb T^n,t }$ using Lemma~\ref{lem:smoothness_properties} as
\begin{align*}
    \log_\partial( W_{\mathbb T^n,t},\mathbb T^n, m)  
    &=
    \log_\partial( (1 -t/T) V_{init} 
    + (t/T)V_{final}, \mathbb T^n,m) \\
    & \le 
    \log 2 + 
    \log_{\partial}(V_{init}, \mathbb T^n, m)  + \log_{\partial}(V_{final} , \mathbb T^n, m).
\end{align*}
By Lemmas~\ref{lem:smoothness_properties} and \ref{lem:f_smoothness_factor}, we get
\begin{align*}
    \log_{\partial}(V_{init} , \mathbb T^n, m) 
    & \ \le \  \log n  +  \log b  +\log_\partial (\Cut(|x_i|) , \mathbb T^n, m)
    \\
  & \ = \
   \log n  +  \log  b   +\log_\partial (\Cut(|x_i|)  , \mathbb R  , m)
   \\
   &
   \ = \
    \log O (bn^{1.5}  m^4).
\end{align*}
Also, by Lemma~\ref{lem:smoothness_compo},
\begin{align*}
    \log_{\partial}(V_{final} , \mathbb T^n, m) 
    & \ \le \
    2 \log m + \log_\partial (\Sat, [0, c+1], m ) + \log_\partial (V, [-L/2,L/2]^n, m ) \\
    & \ \le \
    \log (\Theta(m^2 \cdot m^4\cdot p(m))).
\end{align*}
Combining the two smoothness factors gives
\begin{align*}
    \log_\partial( W_{\mathbb T^n,t},\mathbb T^n, m)  \ \le \ \log(\Theta(bn^{1.5} m^{10} p(m)) \qquad \forall t \in[0,T].
\end{align*}
Therefore, by Lemma~\ref{lem:torus_qubit_equiv}, encoding each spatial dimension with the
\begin{align*}
N \ge \Theta \left( \frac{L^2 \cdot (b^2 n^{11.5} p(2n))^2\cdot   (c+ nb + E_0^6/\sigma^7)^{1.5} }{\epsilon} \right)    
\end{align*}
dimensional quantum 
using $\log N$ qubits
gives that $H_{\mathbb T^n}(t) $ and $H_Q(t)$ are $(E_0,\sigma)$-equivalent.

We have established that $2(\lambda_1(H_\tn(t))+1)\le E_0 $ in Step 2.
Therefore, by Lemma~\ref{lem:gap_approx},
it follows that 
\begin{align*}
        \gap(H_{Q}(t))\ge 0.98g. 
\end{align*}

\paragraph{Step 4:} show $|\lambda_0(H_{\mathbb R^n}(T))-\lambda_0(H_{Q}(T))|\le 0.02\epsilon_0$. 
By Lemma~\ref{lem:pos_based_equivalences},
the condition $V^{-1}_b \subset B^2_r$ gives that
$H_{\mathbb R^n}(T)$ and $H_{\mathbb T^n}(T)$ 
are $(E_0,\sigma)$-equivalent.
Lemma~\ref{lem:ground_energy} implies that 
\begin{align*}
    |\lambda_0(H_{\mathbb R^n}(T) ) -\lambda_0(H_{\mathbb T^n}(T) )| \le \sigma \le 0.01\epsilon_0.
\end{align*}

By Step 3, $H_{\mathbb T^n }(T)$ and $H_Q(T)$ are $\sigma$-equivalent under our choice of parameters.
Lemma~\ref{lem:ground_energy} implies that 
\begin{align*}
    |\lambda_0(H_{\mathbb T^n}(T) ) -\lambda_0(H_{ Q}(T) )| \le 2\sigma \le 0.01\epsilon_0.
\end{align*}
Therefore, by the triangle inequality,
we have 
\begin{align*}
    |\lambda_0(H_{\mathbb R^n} (T)) - \lambda_0(H_Q(T))| \le 0.02 \epsilon_0
\end{align*}

\paragraph{Step 5:} finish the proof.

In the first step of the algorithm, each register measures the ground state with probability $1- O(1/n)$. Therefore $|\Psi_{init}\rangle$
is correctly prepared with probability $\Omega(1)$.

We showed in Step 3 that the spectral gap of $H_Q(t) $ is at least $0.98\cdot3\pi^2/(2(r+1))^2$. 
By Theorem~\ref{thm:adiabatic}, after evolving $|\Psi_{init}\rangle$ under the time-dependent Hamiltonian $H_Q(t)$ for time
\[
  T = O\!\left(
    \frac{\|H_Q(0) - H_Q(T)\|^{2}}{0.01 \cdot (\min_{t \in [0,T]} \gap(H_Q(t)))^3}
  \right) 
   \le  O\!
  \left(
    {r^6 ( c + nb)^2 }  
  \right),
\]
the final state $|\Psi_{final}\rangle$ is $0.01$-close to the ground state of $H_Q(T)$.
Therefore, an off-the-shelf
Hamiltonian simulation algorithm, such as \cite{KSB18},
simulates $H_Q$ in polynomial time.

The Hamiltonian simulation with error 0.01 results in a state that is 0.02-close
to the ground state of $H_Q(T)$, 
which we measure
in the final step of the Algorithm with probability $\ge 96\%$.
The energy we measure deviates from $\lambda_0(h)$ by 
at most $0.02\epsilon_0$ (Step 4 of this proof) plus the error from
the phase estimation, which can be made negligible.
\end{proof}

\subsection{Proof of Theorem~\ref{thm:DL_main}}

\begin{proof}[Proof of Theorem~\ref{thm:DL_main}]
We specify the parameters $a,b,c,E, \sigma,r, L$ in Theorem~\ref{thm:main}, and show that they are polynomials in $n,m, R,1/\epsilon$.
First,
we set 
\begin{align*}
a = 0, \quad r = 2R,\quad L= 3R, \quad E = \Theta(n^2)   , \qquad \sigma = \Theta(\epsilon_0/r^4E_0^{1.5}).
\end{align*}

Condition 2 of Theorem~\ref{thm:main} enforces that
\begin{align*}
    b = \Theta(E^7/\sigma^6) =   \Theta(E^{16} R^{24}/ \epsilon_0^6)  = \Theta(n^{32}R^{24}/\epsilon_0^6)
\end{align*}
and at the same time that $V^{-1}_b \subset B^2_r$,
which we satisfy 
by choosing a sufficiently small $\mu$ in \eqref{eq:DL-potential}.

Note that, for all $x\in\rn$ such that $\|x\|=R$,
there exists $j\in[m]$ such that $a_j\cdot x - b_j \ge 0$,
or equivalently $a_j\cdot 2x - b_j\ge b_j \ge 1$
Therefore, for all $x $ such that $\|x\| = 2R $, we have
\begin{align*}
    V(x) \ge \frac{E}{\mu^6}\Bar_\epsilon(1) = \Theta\left(\frac{n^2 }{\mu^6}\right),
\end{align*}
and for 
\begin{align}\label{eq:DL_sigma_first}
\mu = O \left(\frac{\epsilon_0} {n^{5}R^{4} }\right),    
\end{align}
 we have
$
    V^{-1}_b \subset B^2_r$.

On the other hand, $\mu$ should be also chosen
so that $|\lambda_0(-\Delta^D_\Omega) - \lambda_0(h)| \le O(\epsilon_0)$.
   By \eqref{eq:laplacian_nsquare}, we define the expanded region $\Omega$ with $\epsilon  =  O(\epsilon_0/n^2)$ to obtain
    \begin{align*}
        |\lambda_0(h^D)  - \lambda_0(-\Delta^D_\Omega)|  = \epsilon_0/3.
    \end{align*}
By Lemma~\ref{lem:ground_energy}, we need $(E,O(\epsilon_0))$-equivalence between $h$ and $h^D$ to obtain
\begin{align*}
        |\lambda_0(h)  - \lambda_0(h^D)|  =\epsilon_0/3,
    \end{align*}
where  $E=\Theta(\lambda_0(h)) \le \Theta(n^2)$.
By Lemma~\ref{lem:DL_equivalence},
setting
\begin{align}\label{eq:DL_sigma_second}
    \mu =  O\left(\frac{\epsilon_0}{E (mn^2)^{1/3}}\right) = O\left(\frac{\epsilon_0}{( m n^8)^{1/3}}\right)
\end{align}
gives $(E_0,O(\epsilon_0))$-equivalence between $h$ and $h^D$,
finally yielding
\begin{align*}
    |\lambda_0(-\Delta^D_\Omega) - \lambda_0(h) |  =  \epsilon_0/3.
\end{align*}
Our algorithm estimates $\lambda_0(-\Delta^D_\Omega)$ up to $\epsilon_0/3$ when it estimates $\lambda_0(h)$ up to $\epsilon_0/3$.

To satisfy both \eqref{eq:DL_sigma_first} and \eqref{eq:DL_sigma_second}, we set
\begin{align*}
\mu =  O\left(\frac{\epsilon_0}{n^5  m^{1/3} R^{4}}\right)   . 
\end{align*}

We are only left with $c$, whose upper bound we find by
\begin{align*}
    c \le \max_{x:\|x\|_\infty \le L} V(x) \le \max_{x:\|x\|_2 \le \sqrt n L} V(x) \le \sum_{j\in[m]} \frac{E}{\mu^6} \sqrt n L  = O\left( \frac{ELm\sqrt n}{\mu^6 }\right) 
\end{align*}
Therefore, each parameter in Theorem~\ref{thm:main} is polynomial in $n,m,R, 1/\epsilon_0$.

Also, since $V$ is a sum of $\Bar_\epsilon$, which has a logarithmic smoothness factor by Lemma~\ref{lem:f_smoothness_factor},
the algorithm requires only polynomial time and qubits in the parameters.

\end{proof}

\section{Discussion: Weyl convexity}\label{sec:discussion}

A notion of noncommutative convexity 
is observed in $h$.
The functional $|\psi\rangle \rightarrow\langle \psi|h\psi\rangle $
is convex under the unitaries of translation $T_a$
and modulation $M_b$
\[
T_a(\psi)(x) := \psi(x+a),
\qquad
 M_b(\psi)(x) := e^{i b\cdot x}\,\psi(x),
\]
and more generally, 
under the Weyl operators
\begin{align*}
   W_{a,b} : =  e^{\frac{i}{2}a \cdot b}  M_b T_a
\end{align*}
that are noncommuting.
This notion is analogous to the fact that
a classical convex function is convex under translation
of a fixed point.

We formalize the idea mentioned above that $h[\cdot]$ is convex under Weyl operators.
\begin{definition}[Weyl convexity]
A functional $f:L^2(\rn) \rightarrow \mathbb R$ is \textbf{Weyl-convex} if
for any fixed $\psi \in \mathcal D(f)$, the function
\begin{align*}
(a,b)\rightarrow f(W_{a,b}\psi)    
\end{align*}
is convex over $(a,b) \in \mathbb R^n \times \mathbb R^n$.
\end{definition}

\begin{lemma}
    Given a Schr\"odinger operator $h = -\Delta + V$ with convex $V:\rn \rightarrow\mathbb R$, 
    the energy functional $h[\psi]$ is Weyl-convex.
\end{lemma}
\begin{proof}
Let $p:=-i\nabla$, so $-\Delta=p^2$. The conjugations are well-known to be
\[
T_a^\dagger p\,T_a=p,\quad T_a^\dagger V(x)\,T_a=V(x-a),\qquad
M_b^\dagger p\,M_b=p+b,\quad M_b^\dagger V(x)\,M_b=V(x).
\]
Hence
\[
W_{a,b}^\dagger\,p^2\,W_{a,b}=(p+b)^2.
\]
Since $W_{a,b} $ is unitary, $\|W_{a,b}\psi\|=\|\psi\|$, and therefore
\begin{equation*}
h[W_{a,b} \psi]=\frac{\langle\psi|((p+b)^2 \psi\rangle    +   \langle T_a \psi |   V |T_a  \psi\rangle}{\langle\psi|\psi\rangle}.
\end{equation*}
Set
\begin{equation*}
\bar p:=\left( \frac{\langle\psi | p_j \psi\rangle}{\langle\psi| \psi\rangle}\right)_{j\in[n]} \in\mathbb{R}^n,\qquad
K:=-\frac{\langle\psi|\Delta\psi\rangle}{\langle\psi| \psi\rangle}<\infty.
\end{equation*}
Expanding gives
\begin{equation*}
h[W_{a,b}\psi]=\underbrace{\big(\|b\|^2+2\,b\cdot\bar p+K\big)}_{=:q(b)}
\;+\;\underbrace{\int_{\mathbb{R}^n} V(x-a)\,|\psi(x)|^2\ \text d x}_{=:G(a)}.
\end{equation*}
The function $q(b)$ has Hessian $2I_n\succeq 0$, hence is (strictly) convex in $b$.
For each fixed $x$, the map $a\mapsto V(x-a)$ is convex;
integrating convex functions preserves convexity, so $G(a)$ is convex in $a$.
Therefore $h[W(a,b)\psi]=q(b)+G(a)$ is jointly convex in $(a,b)$.    
\end{proof}

\section{Future directions}\label{sec:future}
\paragraph{Complexity theory.}
 Our current knowledge on the complexity of the Schr\"odinger operators is sparse. 
 As far as we are aware, \cite{ZLLW24} is the only work on the topic.
  They prove that the Schr\"odinger operator is StoqMA-hard for general smooth potential on a bounded Dirichlet domain,
  while simulating Schr\"odinger operators is BQP-hard.
    Is deciding the ground energy of a Schr\"odinger operator with convex potential BQP-complete? 
    What is the complexity of Schr\"odinger operators with the fermionic symmetry?
    What is the complexity of finding the ground energy of a molecule?
    What is the complexity of the convex drum problem?

\paragraph{Optimization of Weyl-convex objectives.}

Let us denote $\widehat x:=(\widehat x_1,\dots,\widehat x_n)$, where $\widehat x_i$ is the multiplicative operator associated with the coordinate function $x_i$ in $\rn$. Similarly, let us denote $\widehat p:=(\widehat p_1,\dots,\widehat p_n)$, where $\widehat p_i:=U_{\mathcal F}^\dagger \widehat x_i U_{\mathcal F}$ for the Fourier transformation unitary $U_{\mathcal F}$ on $L^2(\rn)$.

In this paper, we gave a rudimentary algorithm for computing
\begin{align*}
    \min_{\|\psi\| =1} \langle\psi| \  \|\widehat p\|^2 + C(\widehat x)  \ |\psi\rangle,
\end{align*}
for a convex $C$. A natural extension is to find an algorithm for 
\begin{align*}
    \min_{\|\psi\| =1} \langle\psi| \  C_1(\widehat p) + C_2(\widehat x)  \ |\psi\rangle,
\end{align*}
where $C_1,C_2$ are convex. A step further is to consider
\begin{align*}
    \min_{\|\psi\| =1} \langle\psi| \  C_3(\widehat p,\widehat x)  \ |\psi\rangle,
\end{align*}
for a Weyl-convex $C_3$.

A more ambitious goal is to construct a framework that parallels classical mathematical optimization.
For instance, we can aim to solve constraint problems
\begin{align*}
    \min_{\|\psi\| =1} & \langle\psi| \  C(\widehat p,\widehat x)  \ |\psi\rangle, \\
\text{subject to } &\langle\psi| \  D_i(\widehat p,\widehat x)  \ |\psi\rangle \le a_i \quad \forall i\in[m]. 
\end{align*}
that hopefully are as useful and versatile as linear and semidefinite programming.

\paragraph{Calculus of variations.}
The optimization problems we solve in this paper are 
instantiations of the calculus of variations~\cite{Kot14-CoV}.
We propose the calculus of variations as a venue for
exponential quantum speedups.

A typical problem is
\begin{align*}
    \text{minimize}\qquad  J[u]&:=\int_{x\in\Omega} L(x,u(x),\nabla u(x))\ \text d x \\
    \text{subject to}\qquad  0 & \  =\int_{x\in\Omega} C_i (x,u(x),\nabla u(x))\ \text d x \qquad \forall i \in[m], \\
    u&: \ \Omega \subset \rn \rightarrow \mathbb R,
\end{align*}
where $L,C_i: \rn \times \mathbb R\times \rn$.

The problem is a good candidate for an exponential quantum advantage,
since classical algorithms, such as finite element methods, require exponential time in $n$.
A quantum computer has an advantage in that 
the exponentially large data $u$ can be efficiently 
stored and manipulated in polynomially many qubits
in $n$.

The $n$-dimensional bubble problem~(the minimal surface problem) 
is a more concrete example, 
where we are to compute
\begin{align*}
    \text{minimize }\qquad  & \int_{x\in\Omega}
    \sqrt {1 + | \nabla u (x)|^2 } \ \text d x
    \\
    \text{subject to }\qquad & u: \Omega \subset \rn \rightarrow\mathbb R \\
  &  u(x)  = b(x) \qquad \forall x \in \partial\Omega,
\end{align*}
for some sunccintly described boundary data $b:\partial \Omega\rightarrow\mathbb R$.
Note that this problem is fundamentally different from the above Weyl-convex setting in that the objective is not sesquilinear anymore.
Still, we can find some notion of convexity in the objective,
since a soap bubble surface recovers from 
a displacement.

\paragraph{Dissipation vs. Adiabatic evolution.}
We observe that efficient quantum optimization algorithms 
for finding the ground energy of a quantum Hamiltonian
roughly fall under two categories: dissipation~\cite{CHPZ25, DZPL25}
and adiabatic evolution~\cite{PP14}. 

We conjecture there is a unified framework encompassing both.
Our algorithm solves the Schr\"odinger problem via adiabatic evolution,
but intuitively, a dissipation-based algorithm makes more sense 
for the problem;
it is more natural to imagine a particle in nature
minimizes its mechanical energy via dissipation, rather than by some
adiabatic process.
Furthermore, \cite{CSW25} devices a time-dependent Hamiltonian
in order to simulate quantum friction, alluding that 
friction and adiabatic evolution are interchangeable.

\paragraph{Alone classical, together quantum.}
The kinetic and potential terms, 
each of which can be classically minimized,
are added to give a quantum objective that 
is efficiently optimized by a quantum algorithm.

Does this phenomenon appear in a more general setting?
For instance,
can we add two diagonal Hamiltonians in the 
$X$ and $Z$ basis, each of which is classically efficiently minimized,
to 
get an quantum Hamiltonian that is efficiently minimized by
a quantum algorithm?
Can we add classically approximable diagonal Hamiltonians 
in different basis
to get a quantum approximable Hamiltonian?
Is there a unified framework for these problems?

For a concrete example, the Quantum Max Cut~\cite{GP19} Hamiltonian $\sum_{ij\in E} I-X_iX_j-Y_iY_j-Z_iZ_j $
is the summation of the Max Cut Hamiltonians in the $X,Y,Z$ basis,
each of which can be optimally approximated by 
the Goemans-Williamson algorithm~\cite{GW95}.
A quantum approximiation algorithm for the problem is unknown.

\section*{Acknowledgments}
The author is supported by a KIAS Individual Grant
CG093802 at Korea Institute for Advanced Study.
The author appreciates helpful discussions with Hyukjoon Kwon and Isaac Kim. 

\bigskip 

\bibliography{ref}

\appendix

\section{Omitted proofs: low energy truncation}\label{sec:truncation}
\begin{proof}[Proof of Lemma~\ref{lem:truncation_combined}]
  Since $g\le E$, it suffices to show Lemma~\ref{lem:ground_energy} and Lemma~\ref{lem:second_approx}.
\end{proof}

In addition to proving Lemma~\ref{lem:truncation_combined}, we prove Lemma~\ref{lem:gap_approx}, which relates the spectral gaps of two Hamiltonians.

\begin{lemma}[$\lambda_0$ approximation]
\label{lem:ground_energy}
Let $h_a,h_b$ be self-adjoint linear operators with purely discrete spectrum and $\lambda_0(h_a),\lambda_0(h_b)\ge 0$.
Suppose $h_a$ and $h_b$ are $(E,\epsilon)$-equivalent.
 There exists a constant $c>0$ such that if
$\epsilon \in[0,c]$
and
\begin{align*}
\lambda_0(h_a) + 1 \leq E,    
\end{align*}
 then we have
\[
|\lambda_0(h_a) - \lambda_0(h_b)| \leq 2\epsilon.
\]
\end{lemma}

\begin{proof}
Let $\widetilde{\mathcal D}_a, \widetilde{\mathcal D}_b$ be $(E,\epsilon)$-truncated domains of $h_a,h_b$ that are isomorphic to each other via 
unitary $U:\widetilde{\mathcal D}_a\rightarrow\widetilde{\mathcal D}_b$ 
such that $|h_a[\widetilde\psi]-h_b[U\widetilde\psi]|\le \epsilon$ for all $\psi \in\widetilde{D}_a$.
Suppose $\lambda_0(h_a) = h_a[\psi_0]$ for a normalized $\psi_0 \in \mathcal D(h_a)$.
Because $h_a[\psi_0] < E$, there is an $\epsilon$-truncation of ${\psi_0}$, denoted $\widetilde{\psi}_0\in\widetilde{\mathcal D}_a$. 
By the definition of $(E,\epsilon)$-equivalence, 
we have
$$\lambda_0(h_b)  \le h_b[U\widetilde{\psi}_0]\le h_a[\widetilde{\psi}_0] + \epsilon\le h_a[\psi_0]+2\epsilon = \lambda_0(h_a) +2\epsilon,$$
and therefore 
\begin{align}
    \lambda_0(h_b)- \lambda_0(h_a) \le 2\epsilon. \label{eq:ground_energy_thm}
\end{align}

Similarly, let $\mu_0\in\mathcal D_b$ be a vector such that $h_b[\mu_0]  =\lambda_0(h_b).$
The fact that  $\lambda_0(h_b) \le \lambda_0(h_a)+2\epsilon \le E$, for $\epsilon \in[0,1/2]$, 
allows an $\epsilon$-truncation of $\mu_0$, which we denote $\widetilde{\mu}_0\in\widetilde{\mathcal D}_b$.
 By the definition of $(E,\epsilon)$-equivalence, we have 
\begin{align*}
    \lambda_0(h_a)\le h_a[U^{-1}\widetilde{\mu}_0]\le h_b[\widetilde{\mu}_0] + \epsilon\le h_b[\mu_0]+2\epsilon = \lambda_0(h_b) +2\epsilon,
\end{align*}
and therefore 
\begin{align*}
    \lambda_0(h_a)- \lambda_0(h_b) \le 2\epsilon. 
\end{align*}
Together with \eqref{eq:ground_energy_thm}, we prove the lemma.
\end{proof}

\begin{lemma}[$\lambda_1$ approximation]
\label{lem:second_approx}
Let $h_a,h_b\ge 0$ be a self-adjoint linear operator with purely discrete spectrum and
$\lambda_1(h_a)>\lambda_0(h_a) \ge0$. 
Suppose $h_a$ and $h_b$ are $(E,\sigma)$-equivalent,
and also satisfy the conditions
\begin{align}
&g:= \gap(h_a) \ > \ 0, \nonumber \\
&2({\lambda_1(h_a)} +1) \  \leq \  E,\quad   \label{eq:thm_gap_lambda1_condi} \\
&\epsilon, E^{-1}  \ < \ 1 . \nonumber
\end{align}

Suppose $c>0$ is a sufficiently small universal constant. If $\epsilon \in [0,c]$
and
\begin{align*}
    \sigma =  O \left(\frac{\epsilon g}{E^{1.5}} \right),
\end{align*} then
\begin{align}
|\lambda_1(h_a) - \lambda_1(h_b)| \leq \epsilon . 
\label{eq:thm_gap_ineq}    
\end{align}
\end{lemma}

\begin{proof}
Let $\widetilde{\mathcal D}_a, \widetilde{\mathcal D}_b$ be $(E,\epsilon)$-truncated domains of $h_a,h_b$ that are isomorphic to each other via a unitary $U:\widetilde{\mathcal D}_a\rightarrow\widetilde{\mathcal D}_b$

Let $\mu_0\in \mathcal D(h_b)$ be a normalized vector such that $h_b[\mu_0]  = \lambda_0(h_b) .$ 
From Lemma~\ref{lem:ground_energy}, we know that $h_b[\mu_0]\le \lambda_0(h_a) + 2 \sigma \le E.$
Let $\widetilde{\mu}_0\in \widetilde{\mathcal{D}}_b$ be a $\sigma$-truncation of $\mu_0$, and let  $\widetilde{\alpha}_0:=U^{-1}\widetilde{\mu}_0 \in \widetilde{\mathcal{D}}_a$.
Then $h_a[\widetilde{\alpha}_0]$ is a good approximation of $\lambda_0(h_a)$: 
\begin{align}
h_a[\widetilde{\alpha}_0] \le  h_b[\widetilde{\mu}_0] +\sigma  \le  h_b[{\mu}_0] +2\sigma \le \lambda_0(h_a) + 4\sigma. \label{eq:alpha_energy}    
\end{align}

Let $\psi_0\in \mathcal{D}(h_a)$ be a normalized vector such that  $h_a[\psi_0] = \lambda_0(h_a)$.
We show that $\|\widetilde{ \alpha}_0 - \psi_0\|$ is small.
Since $h_a$ has a purely discrete spectrum, 
its eigenvectors form a complete set of basis. 
Therefore, we can write 
$$\widetilde{\alpha}_0 = \sqrt{1- t^2} \psi_0 + t {\psi}_0^\perp$$
for some $t\ge 0$
and a normalized vector ${\psi}_0^\perp \perp \psi_0$ in $\mathcal D(h_a)$.
We are assuming that $\lambda_0(h_a)$ has a nonzero gap $g$,
so we have $h_a[\psi_0^\perp] \ge \lambda_1(h_a) = \lambda_0(h_a) + g.$ 
Since $\langle \psi_0|h_a|\psi_0^\perp\rangle = 0$, we have 
\begin{align}
 h_a[\widetilde{\alpha}_0] =  (1- t^2) \lambda_0(h_a) + t^2
h_a[{\psi}_0^\perp]
\ge (1- t^2) \lambda_0(h_a) + t^2(\lambda_0(h_a) + g)
= \lambda_0(h_a) + t^2 g.    \label{eq:alpha_energy_lowerbound}
 \end{align}
 Combining \eqref{eq:alpha_energy} with \eqref{eq:alpha_energy_lowerbound}, 
we get
$$t^2 \le 4\sigma/g. $$ 
 Because $1 - \sqrt{1-t^2}  = t^2/({1+\sqrt{1-t^2}})\le t^2$, we have
\begin{align}
  \|\widetilde{\alpha}_0 - \psi_0\| = \sqrt{(1 - \sqrt{1-t^2})^2 +t^2} \le \sqrt{ t^4 + t^2} \le \sqrt{2}t \le 2\sqrt {2 \sigma/g}.      \label{eq:gap_alpha_distance}     
\end{align}

Let $\psi_1\in \mathcal{D}(h_a) $ be a normalized eigenvector such that $h_a[\psi_1] = \lambda_1(h_a)$. 
Since $\lambda_1(h_a) < E$, there exists an $\sigma$-truncation of $\psi_1$, which  we denote $\widetilde{\psi}_1 \in \widetilde{\mathcal D}_a$. 
We know that $\psi_1$ is orthogonal to $\psi_0$.
We now show that $\widetilde{\beta}_1:= U\widetilde{\psi}_1$ is approximately 
orthogonal to $\mu_0$, which in turn will approximately lower bound $\lambda_1(h_b)$ by $h_1[\widetilde{\beta}_1]$:
\begin{align}
    |\langle\widetilde{\beta}_1|\mu_0\rangle| 
    &\;=\;
    |\langle\widetilde{\beta}_1|((|\mu_0\rangle -|\widetilde{\mu}_0\rangle) +| \widetilde{\mu}_0\rangle)|   \nonumber \\
    &\;\le\;
    |\langle\widetilde{\beta}_1| \widetilde{\mu}_0\rangle | +  \|\langle\widetilde{\beta}_1| \| \cdot \| |\mu_0\rangle -|\widetilde{\mu}_0\rangle \|   \nonumber \\
    &\;\le\;
     |\langle\widetilde{\beta}_1|\widetilde{ \mu}_0\rangle|  + \sigma \label{eq:thm_gap_mu_dist}\\
    &\;=\; 
     |\langle U^{-1} \widetilde{\beta}_1|U^{-1}\widetilde{ \mu}_0 \rangle|  + \sigma  \nonumber \\
     &\;=\; 
     |\langle \widetilde{\psi}_1|\widetilde{ \alpha}_0 \rangle|  + \sigma  \nonumber \\
     &\;\le\; 
     |\langle {\psi}_1|\widetilde{ \alpha}_0 \rangle|  + 2\sigma  \label{eq:thm_gap_psi1_dist}\\
     &\;\le\; 
     |\langle {\psi}_1|\psi_0 \rangle|  + 2\sigma + 2\sqrt{2\sigma/g}  \label{eq:thm_gap_psi0_dist} \\
     &\; = \;  2\sigma + 2\sqrt{2\sigma/g} \label{eq:thm_gap_betamu_inner}
\end{align}
Here, \eqref{eq:thm_gap_mu_dist} and \eqref{eq:thm_gap_psi1_dist} 
follow from Cauchy-Schwarz and the norm condition of $\sigma$-truncation, 
and \eqref{eq:thm_gap_psi0_dist} follows from \eqref{eq:gap_alpha_distance}.
Since $\epsilon, E^{-1}, g/E<1$, we have
\begin{align}\label{eq:gaps_trunc_sigma_choice}
    \sigma = O(\epsilon g/E^{1.5}) =O(e \cdot E^{-0.5} \cdot g/E) \le O(\min(\epsilon,\sqrt{\epsilon/E}, \epsilon g/E)). 
\end{align}
It follows that \eqref{eq:thm_gap_betamu_inner} is less than 1.

Therefore, we can write $\widetilde{\beta}_1 = s\mu_0 + \sqrt{1- |s|^2}  {\mu}_0^\perp$ 
for some $s\in \mathbb C$ such that 
$|s|\le 2\sigma + 2\sqrt{2\sigma/g} $ and a normalized vector ${\mu}_0^\perp\in \mathcal D(h_b)$ such that $ \mu_0^\perp \perp \mu_0$. 
Then we have 
\[
\begin{aligned}
 (1-|s|^2)h_b[\mu_0^\perp] \; \le \; |s|^2 \lambda_0(h_b) + (1 - |s|^2)h_b[\mu_0^\perp] 
 &\;=\; h_b[\widetilde{\beta}_1]\\
 &\;=\; h_b[U\widetilde{\psi}_1] \\
 &\;\le\; h_a[\widetilde{\psi}_1] + \sigma \\
 &\;\le\; h_a[\psi_1] + 2\sigma \\
 &\;=\; \lambda_1(h_a) + 2\sigma,
\end{aligned}
\]
and therefore
\begin{align}
\lambda_1(h_b) \le h_b[\mu_0^\perp] 
&\;\le\; \frac{\lambda_1(h_a) +2\sigma}{1-|s|^2}     \nonumber     \\ 
&\;\le\; \frac{\lambda_1(h_a) + 2\sigma}{1- (2\sigma + 2\sqrt{2\sigma/g})^2}  
 \label{eq:lambda1hb_less_than_lambda1ha} \\
 & \ \le \ 2{\lambda_1(h_a)+1}  \nonumber\\
&\;\le\; E  ,     \label{eq:thm_gap_lambda1hb_upperbound}
\end{align}
by \eqref{eq:gaps_trunc_sigma_choice} the condition \eqref{eq:thm_gap_lambda1_condi}.
From Inequalities \eqref{eq:gaps_trunc_sigma_choice} and \eqref{eq:lambda1hb_less_than_lambda1ha},
we have
\begin{align}
 \lambda_1(h_b) - \lambda_1(h_a) 
 &\ \le  \  
 \frac{(2\sigma +2 \sqrt{2\sigma/g})^2\lambda_1(h_a) +  2\sigma}{1-(2\sigma +2 \sqrt{2\epsilon/g})^2} \nonumber\\
&\ \le  \ (2\sigma +2 \sqrt{2\sigma/g})^2E +  
 \frac{ 2\sigma}{1-(2\sigma +2 \sqrt{2\sigma/g})^2} \nonumber\\ 
 &\ \le \ \epsilon
 .  \label{eq:thm_gap_lambda1hb-lambda1ha}
\end{align}

Now we upper bound $\lambda_1(h_a) - \lambda_1(h_b)$ in a similar manner. Let $\mu_1\in \mathcal D(h_b)$ be a normalized eigenvector of $h_b$ 
with the second lowest eigenvalue $\lambda_1(h_b).$ 
From \eqref{eq:thm_gap_lambda1hb_upperbound},
we know that there exists an $\sigma$-truncation of $\mu_1$, 
denoted $\widetilde{\mu}_1 \in \widetilde{\mathcal D}_a$.
Let 
\begin{align*}
\widetilde{\alpha}_1 \;: =\; U^{-1}\widetilde{\mu}_1 \;= \; t'\psi_0 + \sqrt{1-t'^2}\psi_0'^\perp,   
\end{align*}
for some $t'$ such that
$|t'|\in [0,1]$ and $\psi_0'^\perp \in \mathcal D(h_a)$
such that $ {\psi_0'} ^\perp \perp \psi_0$.
By \eqref{eq:gap_alpha_distance} and the norm condition of $\epsilon$-truncation, 
\begin{align}
    |t'| 
    &\;=\;
    |\langle\psi_0| \widetilde{\alpha}_1\rangle|  \nonumber\\
    &\;\le\;        
    |\langle\widetilde{\alpha}_0| \widetilde{\alpha}_1\rangle| + 2\sqrt{2\sigma/g}  \nonumber\\
    &\;=\;        
    |\langle\widetilde{\mu}_0| \widetilde{\mu}_1\rangle| + 2\sqrt{2\sigma/g}  \nonumber\\
    &\;\le\;        
    |\langle{\mu}_0| {\mu}_1\rangle| +2\sigma +  2\sqrt{2\sigma/g}  \nonumber\\
    &\;=\;        
    2\sigma +  2\sqrt{2\sigma/g}.  \nonumber
\end{align}
By similar arguments as before, we get
\begin{align*}
(1 - |t'|^2)\lambda_1(h_a) 
 \le  |t'|^2\lambda_0(h_a) + (1 - |t'|^2)\lambda_1(h_a) 
 =  h_a[\widetilde{\alpha}_1] 
 \le  h_b[\widetilde{\mu}_1] + \sigma 
 \le  \lambda_1(h_b) + 2\sigma
\end{align*}
and,
\begin{align*}
    \frac{\lambda_1(h_b) + 2\sigma}{ 1- (2\sigma + 2\sqrt{2\sigma/g})^2} \ge \lambda_1(h_a).
\end{align*}
By similar calculation that led to \eqref{eq:thm_gap_lambda1hb-lambda1ha},
we get
\begin{align}
 \lambda_1(h_a) - \lambda_1(h_b) 
\le \epsilon,  \nonumber
\end{align}
and combining with \eqref{eq:thm_gap_lambda1hb-lambda1ha} gives
\begin{align}
 |\lambda_1(h_a) - \lambda_1(h_b) |
&\ \le  \ \epsilon.  \nonumber
\end{align}
\end{proof}

\begin{lemma}[gap approximation]\label{lem:gap_approx}
Let $h_a,h_b\ge 0$ be a self-adjoint linear operator with purely discrete spectrum. 
Suppose $h_a$ and $h_b$ are $(E,\sigma)$-equivalent and the following conditions hold:
\begin{align}
&\gap(h_a)  \  \ge g \ >  0, \nonumber \\
&2({\lambda_1(h_a)} +1) \  \leq \  E,\quad    \\
&\epsilon, E^{-1}  \ < \ 1 . \nonumber
\end{align}

Then there exist a universal constant $c>0$ such that, 
if $\epsilon \in \cap[0,c]$ and
\begin{align*}
    \sigma \le \Theta \left(\frac{g^2}{E^{1.5}} \right)
\end{align*}
with a sufficiently small constant factor,
then
\begin{align*}
        |\gap(h_a) - \gap(h_b) | \le 0.01g  .  
\end{align*}
\end{lemma}

\begin{proof}
We have $\sigma = O(g)$, since
$\sigma \le \Theta(g\cdot g/E\cdot 1/E^{0.5})$ and
$g\le \lambda_1(h_a) < E$.
By Lemma~\ref{lem:ground_energy}, 
$\sigma\le \Theta(g)$ implies that 
\begin{align*}
    |\lambda_0(h_a) - \lambda_0(h_b)| \le 0.01g/2.
\end{align*}
By Lemma~\ref{lem:second_approx}, 
$\sigma\le \Theta(g\cdot g/E^{1.5})$ implies that 
\begin{align*}
    |\lambda_1(h_a) - \lambda_1(h_b)| \le 0.01g/2.
\end{align*}
Therefore, the triangle inequality gives 
    \begin{align*}
        |\gap(h_a) - \gap(h_b) | 
        \le |\lambda_0(h_a) - \lambda_0(h_b)| + |\lambda_1(h_a) - \lambda_1(h_b)| 
        \le 0.01g.
    \end{align*}
\end{proof}

\section{Omitted proofs: truncation in position}\label{sec:truncation_in_position_appendix}

\begin{proof}[Proof of Lemma~\ref{lem:pos_based_equivalences}]
Let us denote $L^2(B\subset \Omega_i)$ 
for the space of $L^2$ functions on $\Omega_i\in\{\mathbb R^n, B, \mathbb T^n\}$ that is supported on $B\subset \Omega_i$.
    By Lemma~\ref{lem:position_trunc}, $L^2(B\subset \Omega_i)$ is an $(E,\epsilon)$-truncated domain for $h_i$ for all $i\in\{1,2,3\}$.
    The two $(E,\epsilon)$-truncated spaces $L^2(B\subset \Omega_i)$ and $L^2(B\subset \Omega_j)$
    are identified by the unitary 
    $U: L^2(B\subset \Omega_i)\rightarrow L^2(B\subset \Omega_j)$, such that 
    \begin{align*}
        U\psi(x)  =\begin{cases}
            \psi(x)  \qquad  &\|x\|_2\le r+1 \\
            0 \qquad & \|x\|_2> r+1.
        \end{cases}
    \end{align*}

We have $h_i[\psi] = h_j[U\psi]$ because $V_i(x) = V_j(x)$ when $x \in B.$ 
Therefore, $h_i$ and $h_j$ are $\epsilon$-equivalent.
\end{proof}

\begin{lemma}[Markov's in position]\label{lem:rn_state_approx}
Let $h :=-\Delta + V \in S_{\mathbb R^n}\cup S^D_{B} \cup S_{\mathbb T^n}$ be a Schr\"odinger operator 
on the domain $\Omega\in\{\mathbb R^n, B, \mathbb T^n\}$, where $B:=B^2_{r+1}$ and $L>2(r+1)$. 
Let $\psi \in \mathcal D(h)$ be a normalized state with $h[\psi] \le E$.
 Let $f:\Omega \rightarrow \mathbb [0,1]$ be any measurable function such that 
\begin{align*}
    \begin{cases}
        f(x) \  = \  1 & \text{if }\ \ \|x\|_2 \le r, \\
        f(x) \ \in \ [0,1]  &\text{if }\ \ \| x\|_2 \in [r, r+1],\\
        f(x)  \ =0   &\text{if }  \ \  \|x\|_2\ge r+1.
     \end{cases}
\end{align*}
For a parameter $\sigma \in[0,1/2]$, suppose $V(x)\ge E/\sigma^2
$ for all $x$ with $ \|x\|_2 \ge r$.
Then, 
we have
\begin{align}
    \left\| \psi - f\psi \right\| \  &\le \  \sigma.  \label{eq:(1-f)psi}  
\end{align}
\end{lemma}

\begin{proof}
The proof is by Markov's inequality.
In each case of $\Omega \in \{\mathbb R^n, B, \mathbb T^n\}$, we have
\begin{align*}
   \langle \psi| (-\Delta)|\psi\rangle \ = \  -\int_\Omega \overline{\psi}\nabla \cdot \nabla \psi \ \text d x \ = \ \int_\Omega |\nabla \psi|^2 \ \text d x  \ \ge \  0,
\end{align*}
by integration by parts.
Since $\psi$ is normalized, we have 
\begin{align*}
     E \ =  \ h[\psi] \  
     =  \int_\Omega \overline{\psi}(x) (-\Delta + V(x) ) \psi(x) \ \text d x 
     \ \ge  \int |\psi(x) |^2 V(x) \ \text{d} x  
     \ \ge  \frac{E}{\sigma^2}\int_{\|x\| \ge r} |\psi(x )|^2 \ \text d x  
\end{align*}
Therefore,
\begin{align*} 
\sigma^2    \  \ge 
\  \int _{\|x\| \ge r} |\psi(x)|^2 \ \text d x. 
\end{align*}
Now we can bound
\begin{align*}
    \|\psi - f_r\psi\|^2 
    &
    \ = \ \int \ \ \ \ \ (1 - f_r(x))^2 |\psi(x)|^2 \ \text d x \\
    &
    \ = \  \int_{\| x\| \ge r} (1 - f_r(x))^2 |\psi(x)|^2 \ \text d x \\
    & \ \le \ \int_{\|x\| \ge r} |\psi(x)|^2 \ \text d x \\
    & \ \le  \ \sigma^2.
\end{align*}
\end{proof}

\begin{lemma}[Truncation in position]\label{lem:position_trunc}
Let $h :=-\Delta + V \in S_{\mathbb R^n}\cup S^D_{B} \cup S_{\mathbb T^n}$ be a Schr\"odinger operator 
on the domain $\Omega\in\{\mathbb R^n, B, \mathbb T^n\}$, where $ B:=\{x\in\Omega \ | \ \|x\|_2\le r+1 \} $ and $L>2(r+1)$. 

Suppose $V^{-1}([0,E/\sigma^6])\subset B^2_r$.
Then, the Hilbert space $\widetilde{\mathcal D}:= L^2(B) $ is an $(E,\epsilon)$-truncated domain of $h$, where 
$E>1$, and $\sigma = O(\epsilon/E)$. 
\end{lemma}

\begin{proof}
By Theorem~\ref{lem:completeness} and Lemma~\ref{lem:smoothness_of_eigenvecs},
there exists an orthonormal basis $\{  \  \psi_i \ | \ h\psi_i = \lambda_i\psi _i, \ i \in \mathbb Z_{\ge 0} \ \}$ 
for the Hilbert space such that $0\le \lambda_0\le \lambda_1 \le \lambda_2 \le \cdots$, and each $\psi_i$ is smooth.
Therefore, for a normalized vector $\psi \in \mathcal D(h)$ with $h[\psi]\le E$, we can write 
\begin{align*}
    \psi = \sum_{i\in\mathbb Z_{\ge 0}} c_i\psi_i
\end{align*}
for $c_i\in\mathbb C$ such that $\sum_i |c_i|^2 = 1$.
We apply the Markov's inequality on $\{|c_i|^2\}_i$ as follows. 
For $\sigma>0$, we have
\begin{align*}
 E\ge   h[\psi]  = \sum_{i \in \mathbb Z_{\ge 0}}  |c_i|^2\lambda_i \ge \frac{E}{\sigma^2}
 \sum_{i: \lambda_i >E/\sigma^2} |c_i|^2 ,
 \end{align*}
which gives us
\begin{align}
   \sigma^2 \ge  \sum_{i: \lambda_i >E/\sigma^2} |c_i|^2 = \|\psi - \psi^\downarrow\|^2, \  \ \ \ \ \ \text{where} \ \ \  \ \  \      \psi^\downarrow:=\sum_{i: \lambda_i \le E/\sigma^2 } c_i\psi_i.
  \label{eq:position_trunc_markov_ortho}
\end{align}
Since the spectrum of $h$ is purely discrete, the number of $i$ such that $\lambda_i\le  E/\sigma^2$ is finite.
Otherwise, there would be an accumulation point of $\lambda_i$'s in the compact set $[0,E/\sigma^2]$,
which is a contradiction
to the fact that purely discrete spectrum accumulates only at $\infty$.
Therefore, $\psi^\downarrow$ is smooth.

To prove the theorem, it is sufficient to prove the following claim: the vector $\widetilde{\psi}\in  \widetilde{\mathcal D}$ is an $\sigma$-truncation 
of $\psi$, where
\begin{align*}
\widetilde{\psi} \ & := \ \frac{f\psi^\downarrow
}{\|f \psi^\downarrow  \|} , \\
f(x)\ &:= \Cut_{r,1}(\|x\|_2).
\end{align*}
The function $\widetilde{\psi}$ is supported on $B$, and $\|f\psi^\downarrow\| \le \|\psi^\downarrow\| \le \|\psi\|$.
Therefore, we have $\psi^\downarrow \in \widetilde{\mathcal D}$.
We show that $\widetilde{\psi}$ satisfies the norm and energy conditions of $\sigma$-trunction in the rest of the proof.

We first verify the norm condition of $\sigma$-truncation.
By the triangle inequality,
\begin{align}
  \|  \psi -  \widetilde{\psi}  \| 
  &\ \le\ \|\psi -f\psi^\downarrow\| + \|f\psi^\downarrow     -f\psi^\downarrow
  /\|f\psi^\downarrow \|\|\ \nonumber \\
  &\ =\ \|\psi -f\psi^\downarrow \| + |\|f\psi^\downarrow\| - 1| \\
  &\ =\ \|\psi -f\psi^\downarrow\| + |\|f\psi^\downarrow\| - \|\psi\|| \nonumber \\
  &\ \le\ 2\|\psi - f\psi^\downarrow\| \\  \label{eq:thm_position_trunc_psitilde_nrom}
   & \  \le \  2(\|\psi - \psi^\downarrow\|  +  \| \psi^\downarrow - f\psi^\downarrow\| 
  \  )    
\end{align}
We have $\|\psi - \psi^\downarrow\| \le \sigma$ due to \eqref{eq:position_trunc_markov_ortho}.
Also, we have $h[\psi^\downarrow] \le h[\psi]<E$, because $h[\psi]$ is a convex combination of 
$h[\psi^\downarrow]$ and $h[\psi - \psi^\downarrow]$, and we know $h[\psi^\downarrow]\le E/\sigma^2 \le  h[\psi - \psi^\downarrow]$ by definition.
By applying Lemma \ref{lem:rn_state_approx} on $\psi^\downarrow/\|\psi^\downarrow\|$,
we get
$\|\psi^\downarrow - f\psi^\downarrow\|/\|\psi^\downarrow\| \le \sigma^3. $ 
Therefore, we verify that $\sigma  =O(\epsilon/E)$ for a small enough constant factor 
satisfies the norm condition of $(E,\epsilon)$-truncation as
\begin{align}
    \|\psi - \widetilde{\psi} \| \le 2(\sigma +  \sigma^3\|\psi^\downarrow\|) \le 2(\sigma +  \sigma^3) = O(\epsilon/E)\le \epsilon. \label{eq:position_based_norm_condi}
\end{align}

Now we prove the energy condition.
Since $\nabla f^2 =2f \nabla f = 0$ on the boundary of $B$,   
we can apply the divergence theorem to $    |\psi^\downarrow|^2\nabla f^2 $ on $B$ and get
\begin{align*}
     0 =\oint_{\partial B} (|\psi^\downarrow|^2\nabla f^2 )\cdot \frac{ x }{\|x\|} \; \text d^{n-1}x 
    &=
    \int_{B} \nabla \cdot( |\psi^\downarrow|^2\nabla f^2)  \; \text d^{n}x  \\
    &=
    \int_{B}  2\Re(\overline{\psi^\downarrow} \nabla \psi^\downarrow) \cdot \nabla f^2 +  |\psi^\downarrow|^2\Delta f^2  \; \text d^{n}x.  
\end{align*}
Since $f$ is 0 outside the ball, 
\begin{align}
        -\int_{\Omega} \ \ \Re(\overline{\psi^\downarrow}\nabla \psi^\downarrow) \cdot \nabla f^2 \ \text d x 
        & \ =\    \nonumber \\
        -\int_{B} \Re( \overline{\psi^\downarrow}\nabla \psi^\downarrow) \cdot \nabla f^2  \ \text d x
        & \ = \ \frac{1}{2} \int_{B}  |\psi^\downarrow|^2\Delta f^2  \; \text d x \nonumber
        \\
        & \ = \ \frac{1}{2} \int_{\Omega} \ \  |\psi^\downarrow|^2\Delta f^2  \; \text d x  \ =  \ \langle \psi^\downarrow |\Delta f^2| \psi^\downarrow\rangle/2  \label{eq:divergence}.
\end{align}
Since $h[\cdot]$ always gives a real number, we have
\begin{align}
  &\ \ \ \ \  \ \|f\psi^\downarrow\| \ h[f\psi^\downarrow] \\ \nonumber
  & \ = \    
  \Re(\|f\psi^\downarrow\| \ h[f\psi^\downarrow] ) \\
  & \ = \ 
  \Re \int    -f \overline{\psi^\downarrow} \Delta (f\psi^\downarrow)  + f\overline{\psi^\downarrow} V f\psi^\downarrow \ \text{d}x  \nonumber
    &\\
    &\ = \ 
    \int
     -f |\psi^\downarrow|^2 \Delta f   
     - 2\Re(f\overline{\psi^\downarrow} \nabla f\cdot \nabla \psi^\downarrow) 
     - \Re( f^2 \overline {\psi^\downarrow} \Delta {\psi^\downarrow}) + f^2\overline {\psi^\downarrow} V\psi^\downarrow \ \text{d}x \nonumber
    &\\ 
    &\ = \ 
    \int
     -f |\psi^\downarrow|^2 \Delta f   - \Re(\overline{\psi^\downarrow}\nabla \psi^\downarrow)\cdot \nabla f^2  - \Re( f^2 \overline {\psi^\downarrow} \Delta \psi^\downarrow) + f^2\overline {\psi^\downarrow} V\psi^\downarrow \ \text{d}x \nonumber
    &\\ 
    &\ = \
     \int 
    |\psi^\downarrow|^2 ( - f\Delta f + \Delta f^2/2)
    +  \Re (f^2  \overline{\psi^\downarrow}(- \Delta  +V) \psi^\downarrow) \ \text{d}x \nonumber
    &\\ 
    &\ = \
    \langle \psi^\downarrow | \|\nabla f\|^2 |\psi^\downarrow\rangle  + \Re\langle f^2\psi^\downarrow|h\psi^\downarrow\rangle ,\label{eq:thm_position_trunc_hcmu}
\end{align}
where the second to the last equation follows from \eqref{eq:divergence}, and the last equation is by applying the identity $\Delta f^2 = 2|\nabla f|^2  + 2f \Delta f $.

We will show that the first term of \eqref{eq:thm_position_trunc_hcmu} is small, 
and the second term is approximately $\langle\psi^\downarrow|h\psi^\downarrow\rangle.$
To upper bound the first term, 
note that $\nabla f $ is 
nonzero only outside $B^2_r$ where we show $\psi^\downarrow$ has a low weight.
We define a multiplication operator $p_r:\mathcal H \rightarrow \mathcal H$ such that for $\phi \in \mathcal H,$
\begin{align*}
(p_r\phi)(x) =      \begin{cases}
        \phi(x)  \  & \  \text{if } \ \|x\|_2 \le r, \\
        0  \  & \  \text{if } \ \|x\|_2> r.
    \end{cases}
\end{align*}
We apply Lemma \ref{lem:rn_state_approx} on $\psi^\downarrow/\| \psi^\downarrow\|$ and $p_r$, and obtain 
\begin{align*}
\|(1-p_r)\psi^\downarrow\|  =      \left\|(1-p_r) \frac{\psi^\downarrow}{ \|\psi^\downarrow\|} \right\| \|\psi^\downarrow\| \le \sigma^3 \|\psi^\downarrow\| \le \sigma^3,
\end{align*}
which is equivalent to
\begin{align*}
    \int_{x: \|x\|_2\ge r} |\psi^\downarrow(x) |^2 \ \text d x  \le \sigma^6.
\end{align*}
Therefore, we have
\begin{align}
 \langle \psi^\downarrow | \|\nabla f\|^2 |\psi^\downarrow\rangle
 =
 \int_{x : \|x\|_2\ge r}  \| \nabla f(x)\|^2 |\psi^\downarrow(x)|^2 \ \text d r   
 \ \le \sigma^6 \max_{x: \|x\|_2\ge r}\left|\nabla f(x) \right|^2   . \label{eq:thm_position_trunc_nablaf_upperbound}
\end{align}
We compute the gradient in the polar coordinate with the radius $t = \|x\|_2$ and get
\begin{align*}
 \max_{\|x\|\ge r}   \|\nabla f(x)\|^2 
 &= 
 \max_{t\ge r} \left(\frac{\text d \Cut_{r,1}(t)}{\text d t}\right)^2
 = 
 \left(\frac{  \Bump (t - r) } {\int_r^{r+1} \Bump ( {y - r}) \ \text  d y } \right)^2
 \le
\left(
\frac{\max_{t \in [0,1]} \Bump(t) }{\int_0^1 \Bump(t) \text d t }  \right)^2 
\\
&= {C},
 \end{align*}
where $C$ is a universal constant. 
Applying this upper bound to \eqref{eq:thm_position_trunc_nablaf_upperbound} gives 
\begin{align}
    \langle\psi^\downarrow|\|\nabla f\|^2|\psi^\downarrow\rangle  \le {\sigma^6 C}{}. \label{eq:thm_position_trunc_first_term} 
\end{align}

Now we show that the second term of \eqref{eq:thm_position_trunc_hcmu} is close to $\langle \psi^\downarrow | h\psi^\downarrow\rangle $.
By the Cauchy-Schwarz inequality, 
\begin{align}
    |\langle f^2\psi^\downarrow|h\psi^\downarrow\rangle -\langle \psi^\downarrow | h\psi^\downarrow\rangle | 
\  = \ 
|(\langle f^2\psi^\downarrow| -\langle \psi^\downarrow |)|h\psi^\downarrow\rangle  |  \  \le \
\|f^2 \psi^\downarrow - \psi^\downarrow\|\|h\psi^\downarrow\| \label{eq:thm_position_trunc_secondterm}
\end{align}
By applying Lemma \ref{lem:rn_state_approx} on $\psi^\downarrow/\|\psi^\downarrow\|$ with respect to $1 -f^2$, we bound the first norm 
\begin{align}
    \|f^2\psi^\downarrow - \psi^\downarrow\| \le  \sigma^3 \| \psi^\downarrow\| \le  \sigma^3. \label{eq:thm_position_trunc_secondterm_firstterm}
\end{align}
The second norm is bounded as:
\begin{align}
 &\|h\psi^\downarrow\|       
 \ =\    \sqrt{\sum_{i: \lambda_ i \le E/\sigma^2}\lambda_i^2 |c_i|^2}  
 \le 
     \frac{E}{\sigma^2}. \label{eq:thm_position_trunc_secondterm_secondterm}
\end{align}

Plugging \eqref{eq:thm_position_trunc_secondterm_firstterm}, and \eqref{eq:thm_position_trunc_secondterm_secondterm}
into \eqref{eq:thm_position_trunc_secondterm} gives
\begin{align*}
    |\langle f^2\psi^\downarrow|h\psi^\downarrow\rangle -\langle \psi^\downarrow | h\psi^\downarrow\rangle | 
\  \le \
E\sigma
\end{align*}
Since $\langle \psi ^\downarrow|h\psi^\downarrow\rangle$ is real, we have 
\begin{align*}   
    |\Re\langle f^2\psi^\downarrow|h\psi^\downarrow\rangle -\langle \psi^\downarrow |h\psi^\downarrow\rangle| \le  |\langle f^2\psi^\downarrow|h\psi^\downarrow\rangle -\langle \psi^\downarrow | h\psi^\downarrow\rangle | ,
\end{align*}
and therefore
\begin{align}
    \Re\langle f^2\psi^\downarrow|h\psi^\downarrow\rangle
    \ \le \ & \langle \psi^\downarrow|h\psi^\downarrow\rangle +      \sigma E \nonumber
    \\
    \ = \ &  h[\psi^\downarrow]\|\psi^\downarrow\|^2 +      \sigma E \nonumber
    \\
    \ \le \ &  h[\psi] +      \sigma E . \label{eq:thm_position_trunc_second_term}
\end{align}
By \eqref{eq:thm_position_trunc_hcmu}, \eqref{eq:thm_position_trunc_first_term},  
and \eqref{eq:thm_position_trunc_second_term},  we get 
\begin{align*}
  \|f\psi^\downarrow\| \ h[f\psi^\downarrow]  
    &\ \le \
     h[\psi] +      \sigma E  + {\sigma^6 C}{} 
     .
\end{align*}
The triangle inequality implies
\begin{align*}
    \|f\psi^\downarrow\| \ \ge \ \|\psi^\downarrow\|   -\|(1 - f) \psi^\downarrow\| \ \ge \  1 - \sigma  - \sigma^3.
\end{align*}
Therefore, we have
\begin{align*}
h[\widetilde{\psi}]  
=    
h[f\psi^\downarrow] 
 \le 
\frac{   h[\psi] +      \sigma E+ {\sigma^6 C}{} }{  1 - \sigma  - \sigma^3} \le h[\psi] + \epsilon 
\end{align*}
for $\sigma = O(\epsilon/E) $.

\end{proof}

\section{Omitted proofs: truncation in frequency}\label{sec:truncation_in_frequency_appendix}

In this section, we show that $h_\tn$ and $h_Q$ are equivalent.
We first show that 
\begin{align*}
     {\widetilde{\mathcal D}}_{\mathbb {T}^n,f}  &:= \text{span}\left\{ \  \omega_k \ | 
       \ k\in \mathbb Z^n, \ \|k\|_2\le  K  \right\}  \\
     {\widetilde{\mathcal D}}_{Q}  &:= \text{span}\{ \ |p_k\rangle \ | \  k\in \mathbb Z^n,\|k\|_2\le  K \}, 
\end{align*}
are $(E,\epsilon)$-truncated domains of $h_{\mathbb T^n}$ (Lemma~\ref{lem:freq_trunc_torus}) and $h_Q$ (Lemma~\ref{lem:frequency_trunc_qubit}) respectively,
where
\begin{align*}
      \omega_k(x) \ :=& \  \frac{1}{\sqrt {L^n}}e^{i2\pi k\cdot x/L} \ \ \ \text{for } x \in \mathbb T^n, \\
      |p_k\rangle   &\ 
:= \frac{1}{\sqrt{N^n}} \sum_{y \in\mathcal N^n} e^{i 2\pi k\cdot y/N } | y\rangle 
\end{align*}
are the Fourier basis of the respective space.
They are identified by the unitary $U:{\widetilde{\mathcal D}}_{\mathbb {T}^n,f}\rightarrow {\widetilde{\mathcal D}}_{Q}$
such that $ U\omega_k  = |p_k\rangle.$

\begin{proof}[Proof of Lemma~\ref{lem:torus_qubit_equiv}]
By Lemma~\ref{lem:freq_trunc_torus} and Lemma~\ref{lem:frequency_trunc_qubit},
\begin{align*}
     {\widetilde{\mathcal D}}_{\mathbb {T}^n,f}  &:= \text{span}\left\{ \  \omega_k \ | 
       \ k\in \mathbb Z^n, \ \|k\|_2\le  K  \right\},  \\
     {\widetilde{\mathcal D}}_{Q}  &:= \text{span}\{ \ |p_k\rangle \ | \  k\in \mathbb Z^n,\|k\|_2\le  K \} 
\end{align*}
 are $(E,\epsilon)$-truncated domains of $h_{\mathbb T^n}$ and $h_Q$, respectively,
provided that 
\begin{align*}
    K  =  \Omega\left(\frac{L (E+V_{\max})^{3/2}) }{ \epsilon }\right).
\end{align*}
By Lemma~\ref{lem:torus_qubit_energy_approx},
the unitary $U:\widetilde{\mathcal D}_{\mathbb {T}^n,f}\rightarrow  {\widetilde{\mathcal D}}_{Q}$
defined by $U \omega_k= |p_k\rangle$ 
satisfies $|h[\psi] - h[U\psi]|\le \epsilon$ for all $\psi\in\widetilde{\mathcal D}_{\mathbb {T}^n,f}$,
provided that
    \begin{align*}
     N\ge  4 K \ge 16\left(\frac{L  p(2n)  }{2\pi} \right)^2
        \frac{1}{e^{1/n}}.
    \end{align*}
Therefore,
\begin{align*}
    N = 4K  = \Omega\left(  L^2 (E + V_{\max})^{1.5}(p(2n))^2/\epsilon \right)  \ge    \Omega\left( \max \left\{  {L (E+V_{\max})^{1.5}) }/{ \epsilon } , ({L  p(2n)  } )^2 /{e^{1/n}} \right\}   \right) 
\end{align*}
    gives that $h_{\mathbb T^n}$ and $h_Q$ are $(E,\epsilon)$-equivalent.
\end{proof}

\begin{lemma}[Frequency truncation on torus]\label{thm:frequency_trunc_torus}\label{lem:freq_trunc_torus}
Let $\mathbb T^n :=\mathbb R^n/L\mathbb Z^n $ be the $n$-dimensional torus of length $L$, 
$h = -\Delta + V$ be a Schr\"odinger operator 
on $\mathbb T^n$ such that $V(\mathbb T^n ) \subset [0,V_{\max} ]$ for some $V_{\max}>1$.
Then 
\begin{align*}
     &{\widetilde{\mathcal D}}_{\mathbb {T}^n,f}  := \text{span}\left\{ \  \omega_k \ | 
       \ k\in \mathbb Z^n, \ \|k\|_2\le  K  \right\} 
\end{align*}
is an $(E, \epsilon)$-truncated domain of $h$, where
\begin{align*}
    K \ =& \  \Omega\left( \frac{L (E+V_{\max})^{3/2})}{\epsilon}  \right).
\end{align*}
\end{lemma}

\begin{proof}
Let $\psi \in \mathcal D(h)$ be a normalized vector with energy $h[\psi]  \le E$. 
We show that there exists an $(E,\epsilon)$-truncation of $\psi$ in $\widetilde{ \mathcal D} _{\mathbb T^n, f}$.
Since $\psi$ is also in $L^2(\mathbb T^n)$, $\psi$ has the Fourier decomposition
\begin{align*}
    \psi = \sum_{k\in \mathbb Z^n} \widehat{\psi}_k\omega_k.
\end{align*}
Since
$\langle \psi |V|\psi\rangle \ge 0$, we have 
\begin{align}
   E \ge  \langle \psi | (-\Delta)\psi\rangle 
   = 
   \sum_{k \in \mathbb Z^n} \frac{4\pi^2\|k\|_2^2 \widehat{\psi}_k}{L^2} \langle\psi|  \omega_k \rangle  
   = 
   \sum_{k \in \mathbb Z^n} \frac{4 \pi^2 \|k\|^2 }{L^2}  | \widehat{\psi}_k |^2 .  \label{eq:thm_frequency_trunc_kinetic}
\end{align}
Given the parameter $K$, we define low- and high-frequency vectors
\begin{align*}
\psi^\downarrow \  &: = \sum_{\|k\|\le K} \widehat{\psi}_k\omega_k   \in \widetilde{\mathcal D}_{\mathbb T^n, f} , \\
\psi^\uparrow \  &: = \sum_{\|k\| > K} \widehat{\psi}_k\omega_k    = \psi - \psi^\downarrow.
\end{align*}
Since $\{\omega_k\}_k$ forms a complete orthonormal basis set in $L^2(\mathbb T^n)$, 
Parseval's theorem gives $ \sum_{k \in \mathbb Z^n} |\widehat{\psi}_k |^2  = 1$.
By the Markov's inequality on $|\widehat \psi_k|^2$, we obtain
\begin{align}
\|\psi^\uparrow\|^2 
\ = \  \|\psi^\downarrow - \psi   \|^2
\ = \ \sum_{\|k\| \ge K} |\widehat{\psi}_k |^2  \ \le \ \frac{EL^2}{4\pi^2 K^2} :=\sigma^2    \label{eq:thm_frequency_trunc_sigma}
\end{align}
where $\sigma \ge 0.$
We claim that the vector $\widetilde{\psi}:=\psi^\downarrow/\|\psi^\downarrow\|$ is an $\epsilon$-truncation of $\psi.$ 
For the norm condition:
\begin{align}
    \left\|\frac{\psi^\downarrow}{\|\psi^\downarrow\|} - \psi \right\| 
     & \  \le \ 
    \left\|\frac{\psi^\downarrow}{\|\psi^\downarrow\|} - \psi^\downarrow \right\| + \| \psi ^\downarrow  - \psi \| \nonumber 
    \\
    & \ = \ 
    \left|\| \psi\| -{\|\psi^\downarrow\|} \right| +  \| \psi ^\downarrow  - \psi \| \nonumber 
    \\
    &  \ \le \  2 \| \psi ^\downarrow  - \psi \|  \nonumber
    \\
    &\  \le \  2\sigma.
    \label{eq:thm_frequency_trunc_norm}
\end{align}
We choose  
$\sigma =O( {\epsilon}/{(E + V_{\max}   )} )$, so that $2\sigma\le \epsilon$.

For the energy condition of $(E,\epsilon)$-truncation, we first bound the Laplacian term.
From \eqref{eq:thm_frequency_trunc_kinetic}, we have 
\begin{align}
\langle \psi^\downarrow |(-\Delta)|\psi^\downarrow\rangle - \langle \psi |(-\Delta)|\psi\rangle  =  - \sum_{k :\|k\|_2 > K} \frac{4\pi^2\|k\|^2 }{L^2}|\widehat{\psi}_k |^2 \le 0. \label{eq:thm_frequency_trunc_torus_diff_kinetic}
\end{align}
In order to upper bound $\langle \widetilde{ \psi}|V|\widetilde{\psi}\rangle  
 - \langle \psi|V|\psi\rangle $, we use the fact
$   \| V|\psi\rangle \|  \ \le \ V_{\max} \  \| \psi  \|.$
By linearity and norm considerations, we have
\begin{align}
   | \langle  \psi^\downarrow|V|{\psi}^\downarrow\rangle  
 - \langle \psi|V|\psi\rangle|
 &
 \  = \ 
 |\langle  \psi^\downarrow|V|{\psi}^\downarrow\rangle  
 - (\langle \psi^\downarrow | + \langle \psi^\uparrow | )V (| \psi^\downarrow \rangle + | \psi^\uparrow \rangle  ) |   \nonumber \\
 &
 \ = \
 | - \langle  \psi^\downarrow|V|{\psi}^\uparrow\rangle  
- \langle  \psi^\uparrow|V|{\psi}^\downarrow\rangle  
- \langle  \psi^\uparrow|V|{\psi}^\uparrow\rangle  |  \nonumber \\
&
\ \le \
V_{\max}  (\|\psi^\downarrow \| \|\psi^\uparrow \| + \|\psi^\uparrow \| \|\psi^\downarrow\| +\|\psi^\uparrow \| \|\psi^\uparrow \| ) \nonumber \\
&
\ \le \
V_{\max}(2\sigma +  \sigma^2). \label{eq:thm_frequency_trunc_torus_diff_potential}
\end{align}
Combining \eqref{eq:thm_frequency_trunc_torus_diff_kinetic} and \eqref{eq:thm_frequency_trunc_torus_diff_potential} gives
\begin{align*}
\langle \psi^\downarrow |h|\psi^\downarrow\rangle - \langle \psi |h|\psi\rangle 
& 
\ =  \ 
\langle \psi^\downarrow |(-\Delta)|\psi^\downarrow\rangle - \langle \psi |(-\Delta)|\psi\rangle  + \langle \psi^\downarrow |V|\psi^\downarrow\rangle - \langle \psi |V|\psi\rangle  \\
&
\ \le \ V_{\max}(2\sigma + \sigma^2).
\end{align*}
By the fact that $\|\psi^\downarrow\| \ge \|\psi\| - \|\psi^\downarrow - \psi\| \ge 1 - \sigma$, we have
\begin{align}
    h[\psi^\downarrow] 
    \ = \ 
    \frac{\langle \psi^\downarrow |h|\psi^\downarrow\rangle}{\| \psi^\downarrow\|^2} 
    &
    \ \le \
    \frac{h[\psi] + V_{\max}(2\sigma + \sigma^2 )}{\|\psi^\downarrow\|^2} \nonumber\\
    &
    \ \le \
    \frac{h[\psi] + V_{\max}(2\sigma + \sigma^2)}{( 1- \sigma)^2} \nonumber\\
    &
    \ = \
    h[\psi]  +  \frac{h[\psi](2\sigma - \sigma^2) + V_{\max}(2\sigma + \sigma^2)}{( 1- \sigma)^2} \nonumber 
    \\
    &
    \ \le \
    h[\psi]  +  \frac{E(2\sigma - \sigma^2) + V_{\max}(2\sigma + \sigma^2)}{( 1- \sigma)^2} 
 \nonumber
 \\
 &
    \ \le \ 
   h[\psi]  +  \frac{2\sigma  + \sigma^2 }{( 1- \sigma)^2} (E + V_{\max}) \nonumber \\
   &
   \ \le \ \epsilon \nonumber,
\end{align}
where the last line is by our choice
$\sigma =O( {\epsilon}/{(E + V_{\max}   )} )$.
By \eqref{eq:thm_frequency_trunc_sigma}, 
we have $\sigma \le L\sqrt{E + V_{\max}}  /( 2\pi K)$.
Therefore, we have that $\widetilde{ \mathcal D} _{\mathbb T^n, f}$ is an $(E,\epsilon)$-truncated domain, when we have
\begin{align*}
K  =\Omega\left(\frac{L (E+V_{\max})^{3/2}) }{ \epsilon }\right).     
\end{align*}  
\end{proof}

\begin{lemma}[Frequency truncation on qubit]\label{lem:frequency_trunc_qubit}
Let $h_Q$ be a qubit Schr\"odinger operator on $N^n$-dimensional space such that $V(x) \in [0,V_{\max} ]$ for all $x\in \mathbb T^n$ and $V_{\max}>1$.
Let $\mathcal D_Q$ be the $N^n$-dimensional Hilbert space and let $|p_k\rangle$ be defined as above. 

Then 
\begin{align*}
     &{\widetilde{\mathcal D}}_{Q}  := \text{span}\{ \ |p_k\rangle \ | \  k\in \mathbb Z^n,\|k\|_2\le  K \}, 
\end{align*}
is an $(E, \epsilon)$-truncation of $\mathcal D_Q$ with respect to $h_Q$, where
\begin{align*}
    K \ =& \  \Theta\left( \frac{L (E+V_{\max})^{3/2} }{\epsilon}  \right)
\end{align*}
for some constant factor.
\end{lemma}

\begin{proof}
Let us assume that an element $|\psi\rangle \in \mathcal D_{Q}$ has energy $h[|\psi\rangle]  \le E$. 
We write $|\psi\rangle$ in the Fourier basis:
\begin{align*}
    |\psi\rangle = \sum_{k\in \mathbb Z^n} \widehat{\psi}_k |p_k\rangle .
\end{align*}
 Because 
$\langle \psi |V|\psi\rangle \ge 0$, we have 
\begin{align}
   E \ge  \langle \psi | (-\Delta_Q)|\psi\rangle 
   = 
   \sum_{k \in \mathbb Z^n} \frac{4 \pi^2 \|k\|^2 }{L^2}  | \widehat{\psi}_k |^2 .  \nonumber
\end{align}
After defining vectors
\begin{align*}
|\psi^\downarrow\rangle \  &: = \sum_{\|k\|\le K} \widehat{\psi}_k |p_k\rangle   \in \widetilde{\mathcal D}_{Q} , \\
|\psi^\uparrow\rangle \  &: = \sum_{\|k\| > K} \widehat{\psi}_k |p_k\rangle    = \psi - \psi^\downarrow,
\end{align*}
the rest of the proof is verbatim from the proof of Theorem~\ref{lem:freq_trunc_torus} and therefore omitted. 
\end{proof}

Our main goal in this section is to show that $(h_\mathbb T, \widetilde{\mathcal D}_{\mathbb T,2})$ and $(h_Q, \widetilde{\mathcal{ D}}_Q)$ are $\epsilon$-equivalent, given that 
$V$ is ``smooth enough'' according to a quantifiable notion of smoothness that we call smoothness factor.
The Fourier decomposition for $V_\mathbb T$ plays a crucial role in doing so.

\begin{lemma}\label{lem:torus_qubit_energy_approx}
Let $V:\mathbb T^n \rightarrow \mathbb R$ be a smooth potential on the torus.
Suppose 
\begin{align*}
h & : = -\Delta +V \qquad \qquad \qquad \qquad \quad  \text{on \ }\mathbb T^n, \\
h_{Q}  & : = \  K_{Q} +V_{Q}  \qquad \qquad \qquad \qquad \text{on \ }\log N\times n \ \text{qubits,} \\
    \widetilde{\mathcal D}_{\mathbb T^n,f} 
    & :=
    \Span  \{\omega_k |   \  \ k \in \mathbb Z^n, \ \| k\|_2 \le K \} \\
    \widetilde{ D}_{Q}
    & :=
    \Span  \{|p_k\rangle  |  \   \ k \in \mathbb Z^n, \ \| k\|_2 \le K \}.
\end{align*}
Furthermore, suppose $\log_\partial (V_{\mathbb T}, \mathbb T^n, m) \le \log p(m)$ for some polynomial $p$,
Define a unitary $U:\widetilde{\mathcal D}_{\mathbb T^n,f}\rightarrow \widetilde{ D}_{Q}$ by $U\omega_k:= |p_k\rangle$.

Then, for all $\psi\in\widetilde{\mathcal D}_{\mathbb T^n,f}$, we have $|h_{\mathbb T^n}[\psi] - h_Q[U\psi] | \le \epsilon$, given tha
\begin{align*}
     N\ge  4 K \ge 16\left(\frac{L  p(2n)  }{2\pi} \right)^2
        \frac{1}{e^{1/n}}.
    \end{align*}
\end{lemma}

\begin{proof}
For a normalized state $\psi = \sum_{k: \|k\|_2 \le K} c_k \omega_k \in \widetilde{\mathcal D}_{\mathbb T^n,f} $ and $|\phi\rangle :=U\psi=  \sum_{k: \|k\|_2 \le K} c_k |p_k\rangle \in \widetilde{\mathcal D}_Q$, we prove that
\begin{align*}
    |h[\psi]  - h_Q[\phi]| \le \epsilon.
\end{align*}

We first compute 
\begin{align*}
    h_{{\mathbb T}}[\psi] 
    \ &= \ 
    \langle \psi|(-\Delta) | \psi \rangle  + \langle \psi | V|\psi\rangle  \\
    \ &= \ 
    \sum_{k:\|k\|_2\le K} |c_k|^2 \frac{4\pi^2 \|k\|_2^2}{L^2} + \sum_{\|k\|,\|k'\|\le K} \overline{c_{k'}} c_k \langle \omega_{k'}|V| \omega_k\rangle \\
    \ &= \ 
    \sum_{k:\|k\|_2\le K} |c_k|^2 \frac{4\pi^2 \|k\|_2^2}{L^2} + \sum_{\|k\|,\|k'\|\le K} \overline{c_{k'}} c_k  \int_{x\in\mathbb T^n} \frac{V(x) e^{i 2\pi (k - k')\cdot x/L}}{L^n}    \  \text{ d} x \\
        \ &= \ 
    \sum_{k:\|k\|_2\le K} |c_k|^2 \frac{4\pi^2 \|k\|_2^2}{L^2} + \sum_{\|k\|,\|k'\|\le K} \overline{c_{k'}} c_k  \widehat{V}_{k' - k}    ,
\end{align*}
and similarly,
\begin{align*}
    h_{Q}[\phi] 
    \ &= \ 
    \langle \phi|K_Q | \phi \rangle  + \langle \phi | V_Q|\phi\rangle  \\
    \ &= \ 
    \sum_{k:\|k\|_2\le K} |c_k|^2 \frac{4\pi^2 \|k\|_2^2}{L^2} + \sum_{\|k\|,\|k'\|\le K} \overline{c_{k'}} c_k \langle p_{k'}|V_Q| p_k\rangle \\
    \ &= \ 
    \sum_{k:\|k\|_2\le K} |c_k|^2 \frac{4\pi^2 \|k\|_2^2}{L^2} + \sum_{\|k\|,\|k'\|\le K} \overline{c_{k'}} c_k  \sum_{y\in\mathbb [N]^n} \frac{V(yL/N) e^{i 2\pi (k - k')\cdot y/N}}{N^n}  .  
    \end{align*}
Therefore, 
\begin{align}
    |h[\psi] - h_Q[\phi]|
    & \ = \  \left |\sum_{\|k\|,\|k'\|\le K} \overline{c_{k'}} c_k  \bigg( \widehat{V}_{k' - k}   -  \sum_{y\in\mathbb [N]^n} \frac{V(yL/N) e^{i 2\pi (k - k')\cdot y/N}}{N^n} \bigg) \right|
    \nonumber
    \\
    & \ \le \  \sum_{\|k\|,\|k'\|\le K} |c_{k'}| |c_k|  \max_{\|k\|,\|k'\| \le K} \left| \widehat{V}_{k' - k}   -  \sum_{y\in\mathbb [N]^n}  \frac{V(yL/N) e^{i 2\pi (k - k')\cdot y/N}}{N^n} \right|  
    \nonumber
    \\
    & \ \le \    \max_{\|k\|,\|k'\| \le K} \left| \widehat{V}_{k' - k}   -  \sum_{y\in\mathbb [N]^n}  \frac{V(yL/N) e^{i 2\pi (k - k')\cdot y/N}}{N^n} \right| 
    \label{eq:torus_qubit_fourier_diff1}
    \\
    & \ \le \    \max_{\|k\| \le 2K} \left| \widehat{V}_{k}   -  \sum_{y\in\mathbb [N]^n}  \frac{V(yL/N) e^{-i 2\pi k\cdot y/N}}{N^n} \right| 
     \label{eq:torus_qubit_fourier_diff2}
    \\
    & \ = \    \max_{\|k\|_2 \le 2K} \left|
    \int_{\mathbb T^n} 
    \frac{V (x)e^{-i2\pi k\cdot x /L} }{L^n}  \ \text d x
    -  \sum_{y\in\mathbb [N]^n}  \frac{V(yL/N) e^{-i 2\pi k\cdot y/N}}{N^n} \right|  \nonumber
\end{align}
where \eqref{eq:torus_qubit_fourier_diff1} is by the Cauchy-Schwarz inequality on $|c_k|$, 
and  \eqref{eq:torus_qubit_fourier_diff2} is by the change of variables $k'-k\rightarrow k.$
The expression in the last line 
is the maximum discrepancy between the integral average and the corresponding Riemann sum average of $V(x)e^{-i2\pi k \cdot x /L} $.
The key observation for bounding this quantity is that the discrepancy vanishes whenever the function has no high-frequency Fourier components.
We consider the decomposition $V = V^\downarrow +  V^\uparrow $, where
\begin{align*}
    V ^\downarrow(x) = \sum_{l: \|l\|_\infty \le K} \widehat {V}_l e^{i2\pi l\cdot x/L}, \\ 
    V ^\uparrow(x) = \sum_{l: \|l\|_\infty > K} \widehat {V}_l e^{i2\pi l\cdot x/L}.
\end{align*}
Then for $k $ such that $\|k\|_2 \le 2K$,
\begin{align}
    \sum_{y\in\mathbb [N]^n} \frac{V^\downarrow (yL/N) e^{-i2\pi k\cdot y /N}}{N^n}
    \ &= \ 
    \sum_{l: \|l\|_\infty \le K } \sum_{y\in\mathbb [N]^n}  \frac{\widehat{V}_l e^{i2\pi l\cdot y/N} e^{-i2\pi k\cdot y /N}}{N^n} 
    \nonumber
    \\
     \ &= \ 
    \sum_{l: \|l\|_\infty \le K } \widehat{V}_l \sum_{y\in\mathbb [N]^n}  \frac{ e^{i2\pi (l-k)\cdot y/N} }{N^n} 
    \nonumber
    \\
     \ &= \
   \sum_{\substack{l:\|l\|_\infty \le K \\ l = k \mod N} } \widehat{V}_k. \nonumber 
    \\
      \ &= \ 
    \widehat{V}_k. \nonumber 
\end{align}
The last equation holds because we enforce $N\ge3K+1$. 
Since $\|k \| _\infty \le \|k\|_2 \le 2K$ and $\|l\|_\infty \le K$, we have $\|k - l\|_\infty\le \|k\| _\infty + \|l\|_\infty \le 3K.$
Therefore, $k=l \mod N$ if and only if $k=l \in \mathbb Z^n$.

We conclude that 
\begin{align*}
     |h[\psi] - h_Q[\phi]| 
         & \ \le \    \max_{\|k\| \le 2K} \left| \widehat{V}_{k}   -  \sum_{y\in\mathbb [N]^n}  \frac{V(y L/N) e^{-i 2\pi k\cdot y/N}}{N^n} \right|  \\
         & \ = \    \max_{\|k\| \le 2K} \left| \widehat{V}_{k}   -  \sum_{y\in\mathbb [N]^n}  \frac{V^\downarrow(yL/N) e^{-i 2\pi k\cdot y/N}}{N^n}  -   \sum_{y\in\mathbb [N]^n}  \frac{ V^\uparrow(yL/N)  e^{-i 2\pi k\cdot y/N}}{N^n}   \right| 
         \\
         & \ = \    \max_{\|k\| \le 2K} \left|  \sum_{y\in\mathbb [N]^n}  \frac{ V^\uparrow(yL/N)  e^{-i 2\pi k\cdot y/N}}{N^n}   \right| 
         \\
         & \ \le \    \max_{\|k\| \le 2K}   \sum_{y\in\mathbb [N]^n}  \frac{ |V^\uparrow(yL/N) | |e^{-i 2\pi k\cdot y/N}|} {N^n}    
         \\
         & \ = \    \max_{y\in[N]^n} V^\uparrow(yL/N)  \  \le\   \max_{x\in[0,L]^n} V^\uparrow(x)  
         \\
          & \ = \    \max_{x\in[0,L]^n} \sum_{l:\|l\|_\infty >K}  \widehat{V}_k e^{i2\pi k\cdot x/L} \\
          &\  \le \ \sum_{l:\|l\|_\infty>K}  |\widehat{V}_k|.
\end{align*}

We write $x = (x_1,\dots,x_n) = (x_j, x_{ j*}) \in [0,L]^n$, $l = (l_1,\dots,l_n) = (l_j, l_{j*}) \in \mathbb Z^n$ where $x_{ j*} := (x_1,\dots, x_{j-1},x_{j+1},\dots, x_n)$, and $l_{j*}$ is defined similarly. 
Then 
\begin{align}
    \widehat V_l 
    &= 
    \frac{1}{L^n} \int_{x\in\mathbb [0,L]^n} V(x) e^{-i2\pi l\cdot x/L} \ \text d x \nonumber \\
        & = 
        \frac{1}{L^n}\int_{ x_{j*} \in [0,L]^{n-1} } \ e^{-i2\pi l_{j*}  \cdot \  x_{j*} /L}  \int_{x_j \in [0,L]} V(x_j , x_{j*}) e^{-i2\pi l_j x_j /L}  \ \text d x_j  \ \text d x_{j*}  \label{eq:vl_integration} 
\end{align} 
Fubini's theorem lets us integrate the whole expression first in the variable $x_j$ and then in the other variables sequentially.
By repeatedly integrating by parts in $x_j$ and assuming $l_j \ne 0,$ we have, with the notation $\partial_j^m:=\frac{\partial^m}{\partial x_j^m}$, 
\begin{align*}
     &
     \int_{0}^L V(x_j , x_{j*}) e^{-i2\pi l_j x_j /L}  \ \text d x_j 
    \\ 
    =
    & \frac{L}{-i2\pi l_j} V(x_j, x_{j*})e^{-2\pi l_j x_j/L} \Bigg|_0^L 
    - \frac{L}{-i2\pi l_j}\int_{0}^L e^{-i2\pi l_j x_j /L}   {\partial_j V(x_j,x_{j*})}  \ \text d x_j
    \\
    =
    &  \frac{L}{i2\pi l_j}\int_{0}^L  e^{-i2\pi l_j x_j /L} \ {\partial_j V(x_j,x_{j*})} \   \text d x_j
    \\
    =
    &  \left(\frac{L}{i2\pi l_j} \right)^2 \int_{0}^L   e^{-i2\pi l_j x_j /L} \   {\partial_j^2 V (x_j,x_{j*} )} \ \text d x_j 
    \\
    &\vdots
    \\
    = &
    \left(\frac{L}{i2\pi l_j} \right)^m \int_{0}^L   e^{-i2\pi l_j x_j /L} \ \partial_j^m V(x_j,x_{j*}) \  \text d x_j 
\end{align*}
Plugging the result into \eqref{eq:vl_integration},
we get 
\begin{align*}
    \widehat V_l 
    &
    \ = \ 
    \frac{1}{L^n} \left(\frac{L}{i2\pi l_j} \right)^m 
    \int_{ x \in [0,L]^n } 
     e^{-i2\pi l\cdot  x /L} \ {\partial^m_j V(x)} \  \text d x 
    \end{align*}
 for all $m \in \mathbb N$.    
Since $j\in [n]$ is arbitrary, we take $j$ such that $l_j = \|l\|_\infty$ 
and obtain
\begin{align*}
   | \widehat V_l |
    &
    \ \le \ 
    \frac{1}{L^n} \left(\frac{L}{2\pi \|l\|_\infty} \right)^m 
    \int_{ x \in [0,L]^n } 
      | {\partial^m_j V(x)} | \   \text d x
      \\
      &
      \ \le \
 \left(\frac{L}{2\pi \|l\|_\infty} \right)^m 
     \max_{x\in [0,L]^n} | {\partial^m_j V(x)} | 
 \\
      &
      \ \le \
 \left(\frac{L}{2\pi \|l\|_\infty} \right)^m  V_{\max} ^{(m)}
    \end{align*}
 for all $m \in \mathbb N$
where $V^{(m)}_{\max} := \max \{|\partial_j^m V(x)| \ | \  x\in \mathbb T^n, j\in[n]\}$.
Assuming $m\ge n+1$, we have
\begin{align*}
\sum_{l:\|l\|_\infty > K} |\widehat V_l| &\le  V^{(m)}_\text{max} \sum_{\|l\|_\infty > K}   \left(\frac{L}{2\pi \|l\|_{\infty}}\right)^m 
\\
&=  V^{(m)}_\text{max} \sum_{t= K+1}^\infty \sum_{l: \|l\|_\infty = t}   \left(\frac{L}{2\pi \|l\|_{\infty}}\right)^m 
\\
&
\le   V^{(m)}_\text{max}     \sum_{t = K+1}^\infty 2n(2t+1)^{n-1} \left(\frac{L}{2\pi t}\right)^m 
\\
&
\le 2n\cdot  4^{n-1} V^{(m)}_\text{max}   \left(\frac{L}{2\pi }\right)^m  \sum_{t =K+1}^\infty  \frac {1}{t^{m-n+1 }} 
\\
&
\le   4^{n}n V^{(m)}_\text{max}   \left(\frac{L}{2\pi }\right)^m  \frac {1}{(m-n) {K}^{m-n }}, 
\end{align*}
where the last inequality is by integration over $t \in [K,\infty)$. 
We set $m =2n$.
We want the last quantity to be less than $\epsilon,$ or equivalently
\begin{align*}
    n\log 4  + \log V^{(2n)}_{\max} + m \log \left( \frac{L}{2\pi} \right) - n\log K \le \log \epsilon.
\end{align*}
This inequality is satisfied by choosing $K$ so that
\begin{align*}
  &  2 \frac{\log V^{(2n)} _{\max}}{2n}  + \log 4 + 2\log\left(\frac{L}{2\pi}\right) + \log \frac{1}{\epsilon^{1/n}} \\
 \  \le \  &2 \log_\partial(V, \mathbb T^n, 2n)  + \log 4 + 2\log\left(\frac{L}{2\pi}\right) + \log \frac{1}{\epsilon^{1/n}}  \\
 \  \le \  &2 \log p(2n)  + \log 4 + 2\log\left(\frac{L}{2\pi}\right) + \log \frac{1}{\epsilon^{1/n}}  \\ 
   \ \le \ &\log K. 
\end{align*}
Therefore, we set $K =\Theta\left(\left( p(2n)  \frac{L}{2\pi}\right)^2 \frac{1}{\epsilon^{1/n}} \right)$ to get the desired bound
\end{proof}

\section{Omitted proofs: smoothness}\label{sec:smoothness}

The following lemmas state that having a logarithmic smoothness factor is closed under sum and composition.

\begin{lemma}\label{lem:smoothness_properties}
    Let $f,f_k:\mathbb  T^n \rightarrow\mathbb R$ be smooth for $k\in[r]$, and $c>0$. 
    Then we have
    \begin{align*}
       \log_\partial( cf, \mathbb T^n, m)
      &\  \le \   (\log c)^+ \ + \     \log_\partial( f, \mathbb T^n, m) ,
      \\
       \log_\partial(\sum_{i\in[r]} f_i, \mathbb T^n, m)
      &\  \le \   \log r \ + \  \sum_{i\in[r]}     \log_\partial( f_i, \mathbb T^n, m) .
    \end{align*}
\end{lemma}
\begin{proof}
The proof is by manipulating the maximums:
\begin{align*}
           \log_\partial( cf, \mathbb T^n, m)
           =
    &\max_{l\in[m], j\in[n],x \in\mathbb T^n} \frac{(\log c| \partial^l_j  f(x)|)^+}{l }
    \\
     = &
    \max_{l\in[m], j\in[n],x \in\mathbb T^n} \frac{(\log c + \log | \partial^l_j  f(x)|)^+}{l }
\\
     \le & (\log c)^+ + 
    \max_{l\in[m], j\in[n],x \in\mathbb T^n} \frac{(  \log | \partial^l_j  f(x)|)^+}{l }
=  (\log c)^+ + \log_\partial ( f, \mathbb T^n, m),
\end{align*}
and
\begin{align*}
     \log_\partial( \sum_{k \in[r]} f_k, \mathbb T^n, m)
     \ =\
    &\max_{l\in[m], j\in[n],x \in\mathbb T^n} \frac{(\log | \sum_{k\in [r]} \partial^l_j f_k(x)|)^+}{l }
    \\
    \ \le \ &
    \max_{l\in[m], j\in[n],x \in\mathbb T^n} \frac{(\log (\sum_{k\in [r]} | \partial^l_j  f_k (x)|    )^+  }{l }
    \\
    \  \le \ &
        \max_{l\in[m], j\in[n],x \in\mathbb T^n} \frac{(\log ( r\max_{k\in[r]} | \partial^l_j  f_k (x)|  ) )^+  }{l }
        \\
       \  \le \ &
        \max_{l\in[m], j\in[n],x \in\mathbb T^n} \frac{(\log ( \max_{k\in[r]} | \partial^l_j  f_k (x)|  ) )^+ + \log r }{l }
        \\
      \  \le \ & 
        \log r +  \max_{l\in[m], j\in[n],x \in\mathbb T^n} \frac{\max_{k\in[r]}  (\log  | \partial^l_j  f_k (x)| )^+   }{l }
        \\
      \  \le \ &
       \log r + \max_{k\in[r]} \max_{l\in[m], j\in[n],x \in\mathbb T^n} \frac{(\log  | \partial^l_j  f_k (x)| )^+  }{l }
        \\
       = 
       &
       \log r + \max_{k\in[r]} \log_\partial(f_k, \mathbb T, m) .  
\end{align*}

\end{proof}

\begin{lemma}\label{lem:smoothness_compo}
    Let $g:\mathbb  T^n \rightarrow\mathbb R$ and $f:g(\mathbb T^n) \rightarrow \mathbb R$  be smooth, where $g(\mathbb T^n)$ is the range of $g$. 
    Then we have
    \begin{align*}
       \log_\partial( f\circ g, \mathbb T^n, m)
      \  \le \   2 \log m \ + \  \log_\partial(f, g(\mathbb T^n), m)   \ +  \ 
       \log_\partial  (g, \mathbb T^n, m). 
    \end{align*}
\end{lemma}
\begin{proof}
Fa\`a di Bruno's formula gives $\ \forall x \in \mathbb T^n$ and  $\forall  l \in \mathbb N,$ 
\begin{align*}
    \partial_j^l f(g(x))
    &=
    \sum_{} {\frac {l!}{a_{1}!\,a_{2}!\,\cdots \,a_{l}!}}\cdot f^{(a_{1}+\cdots +a_{l})}(g(x))\cdot \prod _{t=1}^{l}\left({\frac {\partial_j^t g^{}(x)}{t!}}\right)^{a_{t}},
\end{align*}
    and therefore
    \begin{align*}
    |\partial_j^l f(g(x))|
        &\le
    \sum_{} {\frac {l!}{a_{1}!\,a_{2}!\,\cdots \,a_{l}!}} 
\left| f^{(a_{1}+\cdots +a_{l})}(g(x)) \right|
\prod _{t=1}^{l} 
\left|{\frac {\partial_j^t g^{}(x)}{t!}}\right|^{a_{t}},
\end{align*}
where the sum is over $(a_1,\cdots, a_l)\in\mathbb Z_{\ge 0}^l$ such that $\sum_{t=1}^l  t\cdot a_t = l$.
There are at most $l^l$ such $l$-tuples. 
The first factor is at most $l^l$. 
The second factor is at most $\max_{} \{\ | f^{(l')} (y) | \  |  \ l'\in [l] , y\in g(\mathbb T^n)\}$. 
The last factor is at most $\prod_{t=1}^l  | \partial_j^t g(x) |^{a_t}$.
Therefore, we have
\begin{align*}
|\partial^l _j f(g(x)) | &\le l^{2l}   \max_{\substack{ y \in g(\mathbb T^n),\ l' \in[l]}} |f^{(l')} (y)|  \prod_{t=1}^l  | \partial_j^t g(x) |^{a_t}.
\end{align*}
By taking the log of both sides,
\begin{align}
       \log |\partial^l _j f(g(x)) | &\le 2l \log l  + \log  \max_{\substack{ y \in g(\mathbb T^n),\ l' \in[l]}} |f^{(l')} (y)| +  \sum_{t =1}^m \log    | \partial_j^t g(x) |^{a_t}
     \nonumber  \\
       &
       = 2l \log l  +  \max_{\substack{ y \in g(\mathbb T^n),\ l' \in[l]}} \log  |f^{(l')} (y)| +  \sum_{t =1}^m \log    | \partial_j^t g(x) |^{a_t}. \label{eq:logpartial_composition_main_eq}
\end{align}
We bound the second term of \eqref{eq:logpartial_composition_main_eq} as follows: 
\begin{align}
      \ \max_{\substack{ y \in g(\mathbb T^n),\ l' \in[l]}} \log  |f^{(l')} (y)|  \nonumber 
 \  \le &  \
      \max _{\ l' \in[l]} \max_{\substack{ y \in g(\mathbb T^n)}} (\log  |f^{(l')} (y)|)^+ \nonumber
    \\
    =  & \  
    \max_{l' \in [l] } \ l' \   \max_{\substack{ y \in g(\mathbb T^n)}}  \frac{(\log  |f^{(l')} (y)|)^+}{l'}   \nonumber
    \\
    \le  & \  
    \left( \max_{l'' \in [l] } \ l''\right) \ \left(  \max_{\substack{ y \in g(\mathbb T^n) , \ l'\in[l]}}  \frac{(\log  |f^{(l')} (y)|)^+}{l'}  \right)  \nonumber
    \\
    =
    &  \ l  \max_{\substack{ y \in g(\mathbb T^n),\ l' \in[l]}}  \frac{(\log  |f^{(l')} (y)|)^+}{l'}  
    \\
    =
    & \ l \log_\partial ( f, g(\mathbb T^n), l).
    \label{eq:logpartial_composition_f}
\end{align}
Since $\log$ is an increasing function and $\sum_{t=1}^l t \ a_t /l =1$, we bound the third term of \eqref{eq:logpartial_composition_main_eq} as follows:
\begin{align}
     &\ \sum_{t =1}^l \log    | \partial_j^t g(x) |^{a_t} \nonumber \\
     = & \ 
 l \sum_{t =1}^l \frac{t a_t}{l}  \frac{\log   | \partial_j^t g(x)| }{t} \nonumber \\
\le&
\ 
 l \max_{t\in[l]}  \frac{\log   | \partial_j^t g(x)| }{t} \nonumber
  \\
\le &
\ 
 l \max_{t\in[l],y \in \mathbb T^n}  \frac{\log   | \partial_j^t g(y)| }{t} \  \le \ l  \log_\partial (g, \mathbb T^n, l) \label{eq:logpartial_composition_g}
  .
\end{align}
By inserting \eqref{eq:logpartial_composition_f} and \eqref{eq:logpartial_composition_g} into \eqref{eq:logpartial_composition_main_eq} and dividing by $l$,
we get 
\begin{align*}
   \frac{\log|\partial _j^l f(g(x))|}{l} \le 2 \log l + \log_\partial (f, g(\mathbb T^n), l) + \log_\partial (g, \mathbb T^n, l). 
\end{align*}
Since the RHS is always nonnegative, we have
\begin{align*}
   \frac{(\log|\partial _j^l f(g(x))|)^+}{l} \le 2 \log l + \log_\partial (f, g(\mathbb T^n), l) + \log_\partial (g, \mathbb T^n, l). 
\end{align*}
The inequality still holds if we take the maximum over $x \in \mathbb T^n$, $j\in[n]$, and $l \in [m]$ on both sides. Therefore,
\begin{align*}
\log_\partial(f\circ g, \mathbb T^n , m) & =  \\
     \max_{j\in[n],  l \in[m], x\in \mathbb T^n } \frac{(\log|\partial _j^l f(g(x))|)^+}{l} 
     &\le 2 \log m + \max_{l \in[m]} \log_\partial (f, g(\mathbb T^n), l) +  \max_{l \in[m]} \log_\partial (g, \mathbb T^n, l) \\
     &= 2\log m + \log_\partial (f, g(\mathbb T^n),m ) + \log_\partial (g, \mathbb T^n, m).
\end{align*}
\end{proof}

\begin{lemma}\label{lem:f_smoothness_factor}
We have 
\begin{align*}
    \log_\partial (\Bump, \mathbb R,m) &\le \log O(m^4), \\    
     \log_\partial (\Sat_{\alpha,\beta} \mathbb R,m), \log_\partial (\Cut_{\alpha,\beta}, \mathbb R,m),\log_\partial (\Bar_{\beta}, \mathbb R,m)     &\le  \log O(1/\beta + m^4). 
\end{align*}

\end{lemma}
\begin{proof}
The following is a well-known smooth function:
  $$
  f_0(x):= \begin{cases}
      \ \exp \left(-1/x\right) & x> 0
\\ \ 0 &x\le 0.
   \end{cases}
   $$
We use the fact that $\Bump (x) = f_0(x) f_0(1-x)$
and the Leibniz rule
\begin{align*}
    (fg)^{(m)}  = \sum_{j = 0}^m \begin{pmatrix}
    m \\ j
\end{pmatrix} f^{(j)}  g^{(m-j)}
\end{align*}
to get an upper bound
\begin{align}\label{eq:bump_function_nth_diff_summation}
   \max_{x\in \mathbb R} |\Bump^{(m)} (x) | \le   \sum_{j = 0}^m \begin{pmatrix}
    m \\ j
\end{pmatrix}  
\max_{x\in \mathbb R} | f_0^{(i)}(x)| \ 
\max_{x\in \mathbb R} | f_0^{(m-i)}(x) | .
\end{align}
We are left with upper bounding $|f_0^{(j)}|$. Note that for $x>0$, $f_0'(x) =(1/x)^2 e^{-1/x} = t^2 e^{-t}$, where $t:=1/x.$ 

We prove by induction that there exists a polynomial such that $f_0^{(j)}(x)  = p_j(t) e^{-t}$, for all $j \in \mathbb N$.
We showed above the statement for $j=1.$
Suppose the hypothesis is true for $j\in \mathbb N$. Then for $j+1$:
\begin{align*}
    f^{(j+1)}_0 (x)  = \frac{dt}{dx} \frac{d}{dt} (p_{j}(t) e^{-t}) =-t^2\left(p_j'(t) - p_j(t)\right) e^{-t}:=p_{j+1}(t)e^{-t}.
\end{align*}
Therefore, the induction hypothesis is true for all $j \in \mathbb N$.

The maximum of $t^j e^{-t}$, for $t>0$, is $(j/e)^j $ at the maximizer $t = j$. Observe that this quantity increases as a function of $j$ for $j\ge 1.$ For an arbitrary polynomial $p(t):=\sum_{j=0}^{\deg p} a_j t^j,$ define sum of the absolute values of its coefficients $\| p\|_c:= \sum_{j=0}^{\deg p} |a_j|$. Then, if $p$ does not have a constant term,
\begin{align*}
    \max_{t>0}| p_j(t)e^{-t}| \le \sum_{j = 1}^{\deg p_j} |a_j| (j/e)^j\le  \|p\|_c (\deg p /e)^{\deg p}.
\end{align*} Therefore, 
\begin{align*}
   \max_{t >0} \left| f_0^{(j)}(t) \right| \le \|p_j\| \left(\frac{\deg p_j}{e}\right)^{\deg p_j}.
\end{align*}
From the recursive formula $p_{j+1} (t) = t^2(p_j(t)  - p'_j(t))$ with the initial condition $p_1 = t^2 $, we see that $\deg p_j = 2j$ and $p_j$ is without a constant term, for all $j\in \mathbb N.$ From the recursive formula, we also have 
\begin{align*}
 \|p_j\|_c &\le \|t^2 p_j\|_c + \|t^2 p'_j\|_c  \\ 
 &\le 
 \|p_{j-1}\| + 2(j-1)\|p_{j-1}\| 
 \\& = (2j-1)\|p_{j-1}\| 
 \\ & \le (2j-1)(2j-3)\cdots 3\cdot 1\cdot \| p_{1}\|     
\end{align*}
Therefore, $\|p_j\|  <(2j)^j  $, and we have
\begin{align*}
   \max_{x \in\mathbb R} |f_0^{(j)}(x)| =   \max_{t>0} | p(t) e^{-t}| \le (2j^2/e)^j \le (2m^2/e)^m,
\end{align*}
for $0\le j\le m$.
Finally, we can plug this inequality into \eqref{eq:bump_function_nth_diff_summation} and use the binomial identity $\sum_{j=0}^m  \begin{pmatrix}
    m \\ j 
\end{pmatrix} = 2^m$ to get
$    \max_{x\in \mathbb R}|\Bump^{(m)}(x)| \le 2^m(4m^4/e^2)^m,
$
and therefore, 
\begin{align*}
    \log_\partial (\Bump , \mathbb R, m ) 
    &= \sup_{l\in[m], x \in\mathbb R   } \frac{(\log|\Bump^{(l)}(x)|)^+    }{l} 
    \\
    &\le \max_{l\in[m]}\frac{l\log 2 + l\log(4l^4/e^2) }{l}
    \\
    &= \log O(m^4).
    \end{align*}

Since $\Cut_{\alpha,\beta} '(x) = - \Bump(\frac{x - \alpha}{\beta})/\int_{\alpha}^{\alpha + \beta}\Bump(\frac{y-\alpha}{\beta})\text d y  $, we have
\begin{align*}
    \Cut_{\alpha, \beta}^{(m)} (x) =
    - \frac{1}{\beta^{m-1} } 
    \Bump^{(m-1)}
    \big(\frac{x-\alpha}{\beta} \big)  \bigg/ {\int_\alpha^{\alpha + \beta} \Bump\big((y-\alpha)/\beta \big) \  \text {d} y } .        
\end{align*}
Therefore,
\begin{align*}
    \log_\partial(\Cut _{\alpha,\beta},\mathbb R, m) = \sup_{l\in[m],x\in\mathbb R}\frac{( \log|\Cut^{(l)}_{\alpha, \beta}(x)| )^+ }{l}  
    =  \log O(1/\beta + m^4). 
\end{align*}
Since $\Sat_{\alpha, \beta}' = \Cut_{\alpha,\beta}$ and $\Bar_\epsilon' = 1- \Cut_{0,\epsilon}(x) $, 
we apply a similar analysis to obtain
\begin{align*}
    \log_\partial(\Sat_{\alpha,\beta},\mathbb R,m) = \log O(1/\beta + m^4),  \\
    \log_\partial(\Bar_{\epsilon},\mathbb R,m) = \log O(1/\epsilon + m^4).
\end{align*}
\end{proof}

\section{Omitted proofs: the drum problem}\label{sec:DL_appendix}

\begin{lemma}\label{lem:drum_hD-DeltaD}
Suppose $h^D$ is defined as above.
If $    \epsilon = O({\epsilon_0}/{n^2})$, then we have
\begin{align*}
    |\lambda_0(h^D) - \lambda_0(-\Delta^D_\Omega)| \le O(\epsilon_0).
\end{align*}
\end{lemma}
\begin{proof}
By domain monotonicity (Lemma~\ref{lem:domain_monotonicity}), we have $ \lambda_0(h^D)  \le  \lambda_0(-\Delta^D_{\Omega})$,
and by potential comparison (Lemma~\ref{lem:potential_comparison}), we have $ \lambda_0(-\Delta^D_{\Omega'}) \le  \lambda_0(h^D) $, giving
\begin{align*}
   \lambda_0(-\Delta^D_{\Omega'}) \le  \lambda_0(h^D)  \le  \lambda_0(-\Delta^D_{\Omega}).
\end{align*}
Additionally, for every eigenfunction $\psi$ of $\Delta_{\Omega}$, 
we have that $\psi'$ as an eigenfunction of $\Delta_{\Omega'}$, where
 $\psi'(x) = \psi(x/(1+\epsilon)) $.
Therefore, 
\begin{align*}
    \lambda_0(-\Delta^D_{\Omega'})  = \lambda_0(-\Delta^D_\Omega)/(1+\epsilon)^2\ge \lambda_0(-\Delta^D_\Omega) -O(\epsilon\lambda_0(-\Delta^D_\Omega)),
\end{align*}
which gives us  a sandwich
\begin{align*}
     \lambda_0(-\Delta^D_{\Omega})- O(\epsilon\lambda_0(-\Delta^D_\Omega)) \le  \lambda_0(h^D)  \le  \lambda_0(-\Delta^D_{\Omega}).
\end{align*}
Hence, we have 
\begin{align}\label{eq:laplacian_nsquare}
|\lambda_0(h^D)  - \lambda_0(-\Delta^D_\Omega)| \le O(\epsilon\lambda_0(-\Delta^D_\Omega) ) \le  O(\epsilon n^2),    
\end{align}
where the last inequality is from the domain monotonicity (Lemma~\ref{lem:domain_monotonicity}) of
\begin{align*}
B^\infty_{1/2\sqrt n} \subset B^2_1 \subset \Omega.    
\end{align*}
We set 
$    \epsilon = O({\epsilon_0}/{n^2})
$
and obtain 
$    |\lambda_0(h^D) - \lambda_0(-\Delta^D_\Omega)| \le O(\epsilon_0).
$    
\end{proof}

We aim to use Markov's to establish equivalence between $h$ and $h^D$ and eventually claim that $\lambda_0(-\Delta^D_\Omega)$ and $\lambda_0(h)$ are close.
Suppose $x \in \partial \Omega'$. Then, there exists $j\in[m]$ such that 
$a_j \cdot x = b_j(1+\epsilon)$.
Therefore, 
\begin{align}\label{eq:app_boundary_potential}
    V(x) \ge  \frac{3E}{\mu^6}\Bar_\epsilon(a_j\cdot x-b_j) \ge \frac{3E \ \Bar_\epsilon(\epsilon   b_j)}{\mu^6} = \frac{3E}{\mu^6} \qquad \forall\  x \in \partial \Omega.
\end{align}

\begin{lemma}\label{lem:DL_equivalence}
 If $\mu = O(\frac{\epsilon_0}{E (mn^2)^{1/3}})$, then $ h$ and $h^D$ are $(E,O(\epsilon_0))$-equivalent.
\end{lemma}
\begin{proof}
We claim that $L^2(\Omega')$ is an $(E,\epsilon)$-truncated domain for both $h^D$ and $h$.
The claim would prove that $h$ and $h^D$ are $(E,\epsilon)$-equivalent,
since the identity unitary takes $L^2(\Omega')$ to itself, and the restrictions of $h$ and $h^D$ on $L^2(\Omega')$ are identical.

For $h^D$, we have $\mathcal D(h^D)\subset L^2(\Omega')$,
and therefore $L^2(\Omega')$ is an $(E,\epsilon)$-truncated domain for $h^D$.
Therefore, we only need to show that $L^2(\Omega')$ is an $(E,\epsilon)$-truncated domain for $h$.

The derivation is almost the same as in the proof of Lemma~\ref{lem:position_trunc}, except that we use a different cutoff function $f$.
We will refer to the proof of Lemma~\ref{lem:position_trunc}, when the calculation is the same.

Suppose $\psi\in L^2 (\rn)$ has energy $h[\psi]\le E$.
We define 
\begin{align*}
    \psi^\downarrow:=\sum_{i: \lambda_i \le E/\mu^2 } c_i\psi_i,
\end{align*}
where $\psi_i$ is the $i$-th eigenvectors of $h$ with eigenvalue $\lambda_i$.
We define sets 
\begin{align*}
    A   &:=V^{-1}([0,E/\mu^6]) , \\
    B   &:=V^{-1}([0,2E/\mu^6]).
\end{align*}
Then we have $A\subset B\subset \Omega'$, because of \eqref{eq:app_boundary_potential}.

The cutoff function $f:= \Cut_{E/\mu^6,E/\mu^6}\circ V $
satisfies that
\begin{align*}
    f(x) 
    \begin{cases}
        = 1     &  x \in A  \\
        \in[0,1]    & x \in B\setminus A\\
        = 0        & x \notin B.
    \end{cases}
\end{align*}

We claim that 
\begin{align*}
    \widetilde{\psi}:=\frac{f\psi^\downarrow}{\|f\psi^\downarrow\|}
\end{align*}
is an $\mu$-truncation of $\psi$.
The claim implies the lemma,
since $\widetilde{\psi}\in L^2 (\Omega')$.

We first show the norm condition of $\mu$-truncation.
Since $V(x)\ge E/\mu^6$ for all $x$ such that $f(x) = 0$, by Markov's inequality (see the proof of Lemma~\ref{lem:rn_state_approx}), we have
\begin{align*}
    \|\psi^\downarrow - f\psi^\downarrow\|\le \mu^3 \|\psi^\downarrow\|.
\end{align*}
By the same derivation that led to \eqref{eq:position_based_norm_condi}, we have 
\begin{align*}
    \|\psi  -  \widetilde{\psi}\| \le O(\mu) \le O(\epsilon_0).
\end{align*}

For the rest of the proof, we show the energy condition.
By the same calculation that led to \eqref{eq:thm_position_trunc_hcmu}, we have  
\begin{align*}
    \|f\psi^\downarrow \| h[f\psi^\downarrow\|  = \langle \psi^\downarrow | \|\ \nabla f\|^2 |\psi^\downarrow\rangle + \Re\langle f^2 \psi^\downarrow | h\psi ^\downarrow \rangle .
\end{align*}
We upper bound the two terms similarly to the proof of Lemma~\ref{lem:position_trunc}.
We have 
\begin{align*}
    \langle \psi^\downarrow | \|\ \nabla f\|^2 |\psi^\downarrow\rangle 
    =&\int_{x \in \rn} |\psi^\downarrow(x)|^2  \|\nabla f(x)\|^2 \ \text d x \\
    \le 
    & \ 
    \max_{x \notin A} \|\nabla f(x)\|^2 \int_{x\notin A} |\psi^\downarrow(x)|^2  \ \text d x
    \\
    \le
    &
    \   \max_{x \notin A} | \Cut'_{E/\mu^6,E/\mu^6} ( V(x))|^2 \ \|  \nabla V(x)\|^2 \mu^6
    \\
    \le
    &
    \     O\left(\frac{\mu^{12}}{E^2}  \cdot \frac{m ^2 E^2}{\epsilon^2\mu ^{12}}\cdot \mu^6 \right)  
    \\
 = & \ O \left( \frac{m^2 n^4 \mu^{6} }{\epsilon_0^2}\right) \\ 
  \le & \  O(\epsilon^4/ E^4).
\end{align*}
Also, by the calculation that led to \eqref{eq:thm_position_trunc_second_term}, we have
\begin{align*}
    \Re\langle f^2\psi^\downarrow|h\psi^\downarrow\rangle  \le h[\psi]  + \mu E \le h[\psi]  + \epsilon_0.
\end{align*}
Therefore, we have
\begin{align*}
     \|f\psi^\downarrow \| h[f\psi^\downarrow\|  \le h[\psi] + O(\epsilon_0) + O(\epsilon_0^4 E^{4}). 
\end{align*}
Finally, since
\begin{align*}
    \|f\psi^\downarrow\| \ \ge \ \|\psi^\downarrow\|   -\|(1 - f) \psi^\downarrow\| \ \ge \  1 - \mu  - \mu^3,
\end{align*}
we have 
\begin{align*}
h[\widetilde{\psi}]  
=    
h[f\psi^\downarrow] 
 \le 
\frac{   h[\psi] +      O(\epsilon_0) + O(\epsilon_0^4 E^{4}) }{  1 - \mu  - \mu^3} \le h[\psi] + O(\epsilon_0) .
\end{align*}

\end{proof}

\end{document}